\documentclass[12pt]{article}
\usepackage[a4paper]{geometry}
\usepackage{pgfplots}
\pgfplotsset{compat=1.6}
\usepackage{sectsty}
\sectionfont{\normalsize\bfseries}
\subsectionfont{\normalsize\sf}
\usepackage[mathcal]{euscript}

\usepackage{bm}                     %AMS Packages
\usepackage{amsfonts}
\usepackage{amssymb}
\usepackage{amsmath}
\usepackage{amsthm}
\usepackage{graphicx}               %Graphics
\usepackage{natbib}                 %Citation style
\setcitestyle{authoryear,open={(},close={)}}

\usepackage{colortbl}               %Tables
\usepackage{booktabs}
\usepackage{hyperref}
\usepackage{fancybox}
\usepackage{pgfplotstable}

\usepackage{amssymb,amsmath}

\usepackage[utf8]{inputenc}
\usepackage{graphicx,color}
\newtheorem{theorem}{Theorem}
\newtheorem{lemma}{Lemma}
\newtheorem{proposition}{Proposition}

\newcommand{\argmin}{\operatornamewithlimits{argmin}}
\newcommand{\argmax}{\operatornamewithlimits{argmax}}

\allowdisplaybreaks %Varias páginas por ecuación
\definecolor{violet}{rgb}{0.7,0.2,0.6}

\begin{document}

\begin{center}
\large \bf  An RKHS model for variable selection\\ in functional regression
\end{center}
\normalsize

\begin{center}
  Jos\'e R. Berrendero, Beatriz Bueno-Larraz, Antonio Cuevas \\
  Departamento de Matem\'aticas\\
  Universidad Aut\'onoma de Madrid, Spain
\end{center}

\begin{abstract}
\footnotesize {A mathematical model for variable selection in functional regression models with scalar response is proposed. By ``variable selection'' we mean a procedure to replace the whole trajectories of the functional explanatory variables with their values at a finite number of carefully selected instants (or ``impact points''). The basic idea of our approach is to use the Reproducing Kernel Hilbert Space (RKHS) associated with the underlying process, instead of the more usual $L^2[0,1]$ space, in the definition of the linear model. This turns out to be especially suitable for variable selection purposes, since the finite-dimensional linear model based on the selected ``impact points'' can be seen as a particular case of the RKHS-based linear functional model. In this framework, we address the consistent estimation of the optimal design of impact points and we check, via simulations and real data examples, the performance of the proposed method.}
\end{abstract}


%%%%%%%%%%%%%%%%%%%%%%%%%%%%%%%%%%%%%%%%%%%%%%%%%%%%%%%%%%%%%%%%%%%%%%%
%%%%%%%%%%%%%%%%%%%%%%%%%%%%%%%%%%%%%%%%%%%%%%%%%%%%%%%%%%%%%%%%%%%%%%%
%%%%%%%%%%%%%% 				Introduction			  %%%%%%%%%%%%%%%%%%
%%%%%%%%%%%%%%%%%%%%%%%%%%%%%%%%%%%%%%%%%%%%%%%%%%%%%%%%%%%%%%%%%%%%%%%
%%%%%%%%%%%%%%%%%%%%%%%%%%%%%%%%%%%%%%%%%%%%%%%%%%%%%%%%%%%%%%%%%%%%%%%

\section{Introduction: statement of the problem and motivation}

%%%%%%%%%%%%%%%%%%%%%%%%%%%%%%%%%%%%%%%%%%%%%%%%%%%%%%%%%%%%%%%%%%%%%%%
%%%%%%%%%%%%%% 			    The problem 			   %%%%%%%%%%%%%%%%%%
%%%%%%%%%%%%%%%%%%%%%%%%%%%%%%%%%%%%%%%%%%%%%%%%%%%%%%%%%%%%%%%%%%%%%%%

\noindent \textit{The problem under study: variable selection in functional regression}

The study of regression models is clearly among the leading topics in statistics. In particular, these models play a central role in the theory of statistics with functional data, often called Functional Data Analysis (FDA); see \cite{cuevas2014} for an overview on FDA. 

Throughout this paper, we will consider ``functional data'' consisting of independent $X_1=X_1(t),\ldots, X_n=X_n(t)$ observations (trajectories) drawn from a second-order ($L^2$) stochastic process $X=X(t),\, t\in [0,1]$, with continuous trajectories, continuous mean function $m=m(t)$ and covariance function $K(s,t)$. All the involved  random variables are supposed to be defined in a common probability space $(\Omega, {\mathcal A},{\mathbb P})$.

We are interested on regression models with scalar response, of type $Y_i=g(X_i)+\varepsilon_i$, where $g$ is a real function defined on a suitable space ${\mathcal X}$ where the trajectories of our process
are supposed to live and $\varepsilon_i$ are independent errors (and also independent from the $X_i$) with mean zero and common variance $\sigma^2$. 

More specifically, we are concerned with variable selection issues; see, \cite{fan2010selective}, \citet[Sec. 1]{berrendero2016} for additional information and references. Basically, a variable selection  functional method is an automatic procedure that takes a function $\{x(t),\ t\in[0,1]\}$ to a finite-dimensional vector $(x(t_1),\ldots,x(t_{p}))$. The overall idea for variable selection is to choose the variables $x(t_i^*)$   (or, equivalently, the``impact points'' $t_1,\ldots,t_{p}\in[0,1]$; see \cite{kneip2015}), in an ``optimal way'' so that the original functional problem (regression, classification, clustering,...) is replaced with the corresponding multivariate version  of this problem, based on the selected variables.  In the regression problem, this would amount to replace the functional model $Y_i=g(X_i)+\varepsilon_i$ by a finite dimensional version of type
$Y_i=\phi(X(t_1),\ldots,X(t_{p}))+e_i$. Nevertheless, note that still the problem is of a functional nature, since the variable selection process of the $t_i$ involves in principle the full data trajectories.   

\

%%%%%%%%%%%%%%%%%%%%%%%%%%%%%%%%%%%%%%%%%%%%%%%%%%%%%%%%%%%%%%%%%%%%%%%
%%%%%%%%%%%%%% 			 Some motivation			   %%%%%%%%%%%%%%%%%%
%%%%%%%%%%%%%%%%%%%%%%%%%%%%%%%%%%%%%%%%%%%%%%%%%%%%%%%%%%%%%%%%%%%%%%%

\noindent \textit{Some motivation. The classical linear $L^2$-model. Its drawbacks for variable selection purposes}

It is quite natural to assume that the explanatory functional variables $X_i=X_i(t)$ are members of the space $L^2[0,1]$, endowed with the usual inner product $\langle x_1,x_2\rangle_2=\int_0^1x_1(t)x_2(t)dt$, for $x_1,x_2\in L^2[0,1]$. In this setting, the most popular choice for $g$ is, by far, a linear (or affine) operator from $L^2[0,1]$ to ${\mathbb R}$ which leads to a model of type
\begin{equation}
Y_i=\alpha_0+\langle X_i,\beta\rangle_2+\varepsilon_i,\ i=1,\ldots,n, \label{Eq:L2}
\end{equation}
where $X=X(t)$ is the explanatory functional variable, $\alpha_0\in{\mathbb R}$ is the intercept constant and $\beta\in L^2[0,1]$ denotes the slope function. As in the standard multivariate regression model, the aim here is to estimate $\alpha_0$ and $\beta$ in order to be able to make accurate predictions of the response variable $Y$. 

The corresponding theory is outlined in several places; see, e.g., the article \cite{cardot2010} or the books \cite{ferraty2006} and \cite{horvath2012}. The Hilbert structure of the $L^2[0,1]$ space allows us to keep ourselves as close as possible to the usual least square framework in multivariate regression; for example, the projection $P(x)$ of an element $x$ on a subspace $H$ is characterized by the orthogonality condition $\langle x-P(x), a\rangle_2=0$, for all $a\in H$. However, other crucial differences with the finite-dimensional case (mostly associated with the non-invertibility of the covariance operator of the process $X(t)$) makes the functional $L^2$ regression theory far from trivial. Most of these difficulties are intrinsic to the infinite-dimensional nature of the data, so that they cannot be overcome by just replacing $L^2[0,1]$ with another function space. However, when it comes to variable selection applied to linear regression, it would be useful to have the finite dimensional linear model (based on the selected variables) 
\begin{equation}\label{Eq:L2VS}
Y_i=\alpha_0+\sum_{j=1}^{p}\beta_jX_i(t_j)+\varepsilon_i,\ i=1,\ldots,n
\end{equation}
as a particular case of our general model. Notice that \eqref{Eq:L2VS} cannot be established in the $L^2$ framework, since a transformation of type  $x\in L^2[0,1]\mapsto \sum_{i=1}^p\beta_i x(t_j)$ is not a linear continuous functional in $L^2$. In heuristic terms, one would need to look for the regression function $\beta$ in a suitable space, for which a ``finite-dimensional'' model such as \eqref{Eq:L2VS} could make sense. More precisely, we will replace the $L^2[0,1]$ space in which $\beta$ is assumed to live with the Reproducing Kernel Hilbert Space (RKHS), ${\mathcal H}(K)$, associated with $K$. 

The assumed membership of $\beta$ to ${\mathcal H}(K)$ entails some additional restrictions of regularity  on the slope function $\beta$ (when compared to the simple assumption $\beta\in L^2[0,1]$). In any case, some restrictions on $\beta$ appear also in different ways when the classical  $L^2$-model \eqref{Eq:L2} is considered. The reason is that  the space $L^2[0,1]$ is in fact too large from several points of view. Hence, in spite of the advantages commented above, one typically uses penalization or projection methods to exclude extremely rough solutions in the estimation of $\beta$.

Our proposal here, as presented in the next section, aims at reconciling two targets: first, we look for a functional linear model, wide enough to include finite-dimensional versions, such as \eqref{Eq:L2VS}, as particular cases. Second, we would like to achieve such goal with a minimal change in the ``parameter space'' where $\beta$ lives.

\

\noindent \textit{Some related literature}

A quite general RKHS-based approach to the problem of dimension reduction in functional regression been proposed by \cite{hsing2009}. These authors follow the inverse regression methodology to deal with a model of type $Y=\ell (\xi_1,...\xi_d)+\epsilon$
where $\ell$ is a link function and the $\xi_j$ are linear functionals of the explanatory variable $X$, defined in RKHS terms. This pioneering reference shows very clearly the huge potential of the RKHS approach. However, as the authors point out, there are still many aspects not considered in that paper and worth of attention. Variable selection is one of them. In fact, the whole point of the present paper is to show that things become particularly simple when the RKHS machinery is applied to variable selection. A recent use of the RKHS methods in the problem of functional binary classification is developed in
\cite{berrendero2017}.
See also \cite{berlinet2004} and \cite{hsing2015theoretical} for a broader perspective of the applicability of RKHS methods in statistics. 

Other variable selection methods, always aimed at selected the ``best points'' $t_1,\ldots,t_d$ (or the ``best variables'' $X(t_1),\ldots,X(t_p)$) have been proposed as well, with no explicit reference of RKHS tools. Thus, the selection of the ``best impact point'' $t_1$ in a model of type with $p=1$ \eqref{Eq:L2VS} is addressed 
by \cite{mckeague2010}.
Also, different variable selection methods have been suggested by \cite{ferraty2010most} 
and \cite{delaigle2012} for prediction and classification purposes, respectively.

%%%%%%%%%%%%%%%%%%%%%%%%%%%%%%%%%%%%%%%%%%%%%%%%%%%%%%%%%%%%%%%%%%%%%%%
%%%%%%%%%%%%%% 				Notation				   %%%%%%%%%%%%%%%%%%
%%%%%%%%%%%%%%%%%%%%%%%%%%%%%%%%%%%%%%%%%%%%%%%%%%%%%%%%%%%%%%%%%%%%%%%

\

\noindent \textit{Some notation}

A set of possible ``impact'' points  $t_1,\ldots,t_p\in [0,1]$ will be denoted $T$ (sometimes $S$) or $T_p$ when we want to stress the cardinality of $T$.  Also, $X(T_p)$ will stand for $(X(t_1),\ldots,X(t_p))'$. The superindex $^*$ will be used to denote that the points $t_i^*$ are the ``true'' ones, or the ``optimal'' ones according to some criterion.

Given a random variable $Z$ (with finite variance) the notation  $Z_{T_p}$ will refer to the $L^2$-projection of $Z$  on the space spanned by the components of $X(T_p)-m(T_p)$.

If $p^*<p$, the notation $T_{p^*}\prec T_p$ will indicate that all the points in $T_{p^*}$ belong also to $T_p$.

Finally, as usual in statistics, we use a hat to denote the estimated quantities (or the predicted variables). For instance, $\widehat{T}_p$ will denote a data-driven estimator of $T_p$ and   $\widehat{Y}_{\widehat{T}_p}$ will stand by the corresponding (fully data-driven) prediction of the response $Y_{T_p}$. The halfway notation $Y_{\widehat{T}_p}$ will represent the  projection of the response variable onto the space spanned by the marginal variables indexed by the estimated points $\widehat{T}_p$. 

\

\noindent \textit{Organization of the paper}

In Section \ref{Sec:RKHS} we introduce and motivate (in population terms) our variable selection 
procedure. The asymptotic properties of the empirical version (when the parameters are estimated) are considered in Section \ref{Sec:SampleProperties}. The problems associated with the choice of the number $p$ of selected variables are analyzed in Sections \ref{Sec:NumberP} and \ref{Sec:pnot}. The empirical results (simulations and real data examples) are presented in Section \ref{Sec:experiments}. Some technical proofs are included in Section \ref{Sec:Proofs}.
%%%%%%%%%%%%%%%%%%%%%%%%%%%%%%%%%%%%%%%%%%%%%%%%%%%%%%%%%%%%%%%%%%%%%%%
%%%%%%%%%%%%%%%%%%%%%%%%%%%%%%%%%%%%%%%%%%%%%%%%%%%%%%%%%%%%%%%%%%%%%%%
%%%%%%%%%%%%%% 		Regression model and estimation   %%%%%%%%%%%%%%%%%%
%%%%%%%%%%%%%%%%%%%%%%%%%%%%%%%%%%%%%%%%%%%%%%%%%%%%%%%%%%%%%%%%%%%%%%%
%%%%%%%%%%%%%%%%%%%%%%%%%%%%%%%%%%%%%%%%%%%%%%%%%%%%%%%%%%%%%%%%%%%%%%%

\section{An RKHS-based linear model suitable for variable selection}\label{Sec:RKHS}

Our choice of the ambient space for the slope function $\beta$ is, in some sense, ``customized'' for the problem at hand, since we will consider the Reproducing Kernel Hilbert Space (RKHS) associated with the process $\{X(t),\ t\in[0,1]\}$.

 The theory of RKHS goes back to the 1950's; see \citet[Appendix F]{janson1997} for details and references. It has found a surprisingly large number of applications in different fields, including statistics, see \cite{berlinet2004}.

\

%%%%%%%%%%%%%%%%%%%%%%%%%%%%%%%%%%%%%%%%%%%%%%%%%%%%%%%%%%%%%%%%%%%%%%%
%%%%%%%%%%%%%% 					RKHS				   %%%%%%%%%%%%%%%%%%
%%%%%%%%%%%%%%%%%%%%%%%%%%%%%%%%%%%%%%%%%%%%%%%%%%%%%%%%%%%%%%%%%%%%%%%

\subsection{RKHS spaces in a nutshell } \label{Sec:TeoRKHS}

Before establishing our RKHS-based regression model we need a minimal background on RKHS.  Our starting point will be our underlying $L^2$-process $X=\{X(t)\, t\in[0,1]\}$ with a continuous strictly positively definite covariance function $K(s,t)$ and a continuous mean function $m(t)$.

We first introduce an auxiliary space, associated with $K$, which we will denote by  ${\mathcal H}_0(K)$. It is defined by the set of all finite linear combinations of type $\sum_i^na_iK(s,t_i)$, that is,
$$
{\mathcal H}_0(K):=\{f:\ f(s)=\sum_{i=1}^na_iK(s,t_i),\ a_i\in{\mathbb R},\ t_i\in[0,1],\ n\in{\mathbb N}\}.
$$
In such space we define an inner product $\langle\cdot,\cdot\rangle_K$ by
$
\langle f,g\rangle_K=\sum_{i,j}\alpha_i\beta_jK(s_j,t_i),
$
where $f(x)=\sum_i\alpha_i K(x,t_i)$ and $g(x)=\sum_j\beta_jK(x,s_j)$.

Now, the RKHS associated with $K$, denoted by ${\mathcal H}(K)$,  is defined as the completion of ${\mathcal H}_0(K)$. More precisely, ${\mathcal H}(K)$ is the set of functions $f:[0,1]\rightarrow {\mathbb R}$ obtained as $t$-pointwise limits of Cauchy sequences $\{f_n\}$ in ${\mathcal H}_0(K)$; see \citet[p. 18]{berlinet2004}.

As a conclusion, in heuristic terms, one could say that ${\mathcal H}(K)$ is made of all linear combinations of type $f(s)=\sum_i^na_iK(s,t_i)$ plus all the functions which can be obtained as limits  of them. A natural question is when we can ensure that we have identifiability in this space. It is easy to see that the elements of $\mathcal{H}_0(K)$ have a unique representation in terms of $K$ whenever $K$ is strictly positive definite. For additional details we refer again to \citet[Appendix F]{janson1997}. 

Among the many interesting properties of RKHS spaces, let us especially recall two which will be particularly useful in what follows. 

\

\textit{Reproducing property}. 
$f(t)=\langle f,K(\cdot,t)\rangle_K$, for all $f\in{\mathcal H}(K)$, $t\in[0,1]$.

\

\textit{Natural congruence}.  Denote by ${\mathcal L}_X$, the linear (centered) span of $X$ (i.e. the family of finite linear combinations of type $\sum_i\lambda_i (X(t_i)-m(t_i))$). Let $\overline {\mathcal L}_X$ be the $L^2$-completion of ${\mathcal L}_X$. It is clear that $\overline {\mathcal L}_X$ is a closed subspace of the usual Hilbert space $L^2(\Omega)$ of random variables with a finite second moment; this can be seen as the minimal Hilbert space including the variables $X(t)$. It can be proved, see \citet[Th. 35]{berlinet2004}, that $\Psi_X(\sum_ia_i(X(t_i)-m(t_i))=\sum_i a_i K(\cdot,t_i)$
defines (when extended by continuity) a congruence between $\overline {\mathcal L}_X$ and ${\mathcal H}(K)$. This means that the extension of  $\Psi_X$ is a linear bijective transformation which preserves the inner product. Such congruence is often called \textit{Lo\`eve's isometry}. In explicit terms (see \citet[Lemma 1.1]{lukic2001}), Lo\`eve's isometry between  $\overline {\mathcal L}_X$
and ${\mathcal H}(K)$ can be defined by 
\begin{equation}\label{isometry}
Y\mapsto \Psi_X(Y)=\Psi_X(Y)(t)={\mathbb E}(Y(X(t)-m(t)))
\end{equation}

So, the RKHS space ${\mathcal H}(K)$ is an isometric copy  of $\overline {\mathcal L}_X$. In this sense, it contains the ``Dirac deltas'' $X(t_i)-m(t_i)$ at $t_i$ (as ${\mathcal H}(K)$ contains $K(\cdot,t_i)$). Both spaces $\overline {\mathcal L}_X$ and ${\mathcal H}(K)$ can be identified. 

There is, however, a not-so-nice feature in the RKHS space associated with the process $X(t)$: under very general conditions, this space does not contain, with probability one, the trajectories of the process $X$; see, e.g., \citet[Th. 11]{pillai2007}, \cite[Cor. 7.1]{lukic2001}. This will have some consequences in the formulation of our regression model, as pointed out below.

\

%%%%%%%%%%%%%%%%%%%%%%%%%%%%%%%%%%%%%%%%%%%%%%%%%%%%%%%%%%%%%%%%%%%%%%%
%%%%%%%%%%%%%% 				  RKHS Model				%%%%%%%%%%%%%%%%%%
%%%%%%%%%%%%%%%%%%%%%%%%%%%%%%%%%%%%%%%%%%%%%%%%%%%%%%%%%%%%%%%%%%%%%%%

\subsection{The RKHS functional regression model}\label{Sec:Model}

We propose to replace the standard $L^2$ functional regression model \eqref{Eq:L2} with the following RKHS counterpart
\begin{equation}
\label{Eq:RKHSmodel}
Y_i=\alpha_0+\langle X_i,\beta\rangle_K+\varepsilon_i,\ i=1,\ldots,n
\end{equation}
where $\beta\in {\mathcal H}(K)$ and $\langle\cdot,\cdot\rangle_K$ denotes the inner product in ${\mathcal H}(K)$. 

Since the estimation of the intercept term $\alpha_0$ is straightforward from those of $\beta$ and $m$, we will assume, without loss of generality, that $\alpha_0=0$ in what follows. 

As mentioned at the end of the previous subsection, it is important to keep in mind that the trajectories $X_i$ of the process $X$ do not belong to ${\mathcal H}(K)$. Thus, the expression $\langle X_i,\beta\rangle_K$ has no direct meaning, unless it is appropriately interpreted: in what follows, $\langle X_i,\beta\rangle_K$
must be understood as $\Psi_{X_i}^{-1}(\beta)$, where $\Psi_X$ is defined in \eqref{isometry}. Such an interpretation of $\langle X,\beta\rangle_K$ arises in the classical paper by \citet[Th. 7A]{parzen1961}, aiming at different statistical purposes.  This situation is conceptually similar to that arising in the definition of It\^o's integral, so that $\langle X,\beta\rangle_K$ can be interpreted as a stochastic integral. In fact, when $X(t)$ is a standard Brownian Motion  model \eqref{Eq:RKHSmodel} is equivalent to 
\[
Y = \int_0^1 \beta (s) dX(s) +\varepsilon,\ \mbox{with }\beta\in{\mathcal H}(K),
\]
where $\int_0^1 \beta (s) dX(s)$ is It\^o's integral.

%%%%%%%%%%%%%%%%%%%%%%%%%%%%%%%%%%%%%%%%%%%%%%%%%%%%%%%%%%%%%%%%%%%%%%%
%%%%%%%%%%%%%% 			 Variable Selection			%%%%%%%%%%%%%%%%%%
%%%%%%%%%%%%%%%%%%%%%%%%%%%%%%%%%%%%%%%%%%%%%%%%%%%%%%%%%%%%%%%%%%%%%%%

\subsection{Variable selection in the RKHS functional regression model}

Consider the RKHS functional regression model \eqref{Eq:RKHSmodel} introduced in the previous paragraph,
where $\mathbb{E}(\varepsilon)=\mathbb{E}((X(t)-m(t))\varepsilon)=0$ and $\mbox{Var}(\varepsilon)=\sigma^2$.

\

\noindent \textit{Our goal}

Under this model, for fixed $p$, we aim at selecting $p$ values $t_1,\ldots,t_p$ in order to use the $p$ dimensional vector $(X(t_1),\ldots,X(t_p))$ instead of the whole trajectory $\{X(t):\, t\in [0,1]\}$ in our regression problem. Formally, we want to establish a transformation
$$
\{X(t):\, t\in [0,1]\}\mapsto (X(t_1),\ldots,X(t_p)),
$$
which should be ``optimal'' in the sense that the points $(t_1,\ldots,t_p)$ are chosen according to an optimality criterion, oriented to minimize the information loss in the passage from infinite to finite dimension.

In this section, we address this problem at the population level, that is, we assume that the parameters defining the model (the regression function $\beta$, the covariance function $K$ of the process $X$, the mean function $m$ and the variance of the error variable, $\sigma^2$) are known. Of course, the practical implementation will require using suitable estimators of the unknown parameters. This raises several questions concerning the sample behavior of the method which will be addressed in subsequent sections.

\

\noindent \textit{The optimality criterion $Q_1$}

The first obvious question to address in such strategy is the choice of the optimality criterion. We will see that, in fact, different criteria can be used but, fortunately, they are all equivalent. 

One of the basic goals of a functional regression model is to predict the value of the response variable for a given realization of the regressor function. Then, a sensible approach for variable selection is to choose the $p$ points $(X(t_1),\ldots,X(t_p))$ that give the best linear prediction (in the sense of the $L^2$ norm) of $Y$. This implies to find the vector $T_p$ that minimizes the function
\begin{equation}
\label{Eq:Q1}
Q_1(T_p) := \min_{(\beta_1,\ldots,\beta_p)\in\mathbb{R}^p} \|Y - \sum_{j=1}^p \beta_j (X(t_j)-m(t_j))\|^2_2.
\end{equation}

\

\noindent \textit{Where to look for the optimum}

An important technical aspect is the choice of an appropriate subset $\Theta_p\subset [0,1]^p$ to look for the optimum of the continuous function $Q_1$.  This subset must be compact in order to guarantee the existence of the optimum. Moreover, if we want to get a meaningful optimal value of $T_p=(t_1,\ldots,t_p)$ 
we should rule out those points including repeated values in the coordinates $t_i$. To this end,  we will fix an arbitrarily small fixed value $\delta>0$ , and will look for our optimum in the space 
\begin{equation}
\label{Eq:EspTheta}
\Theta_p = \Theta_p(\delta) = \{T_p=(t_1,\ldots,t_p)\in [0,1]^p:\,  t_{i+1}-t_i\geq \delta,\ \mbox{for } i=0,\ldots,p\},  
\end{equation}where $t_0=0$, $t_{p+1}=1$. In practice, the restriction to the subset $\Theta_p$ is not relevant, since we observe the functions in a finite grid, and we can set $\delta>0$ as small as required so that all the points in the grid belong to $\Theta_p$.

The reason for the choice  \eqref{Eq:EspTheta} of $\Theta_p$ is technical, very much in the same spirit of \citet[Eq. (9)]{muller2016}. We need to work on a compact set and, at the same time, to avoid degeneracy problems in the choice of the points $(t_1,\ldots,t_p)$ that could lead to a singular covariance matrix  in $(X(t_1),\ldots,X(t_p))$. 
Other choices are possible for $\Theta_p$. For example, one could think of defining $t_0=0$, $t_{p+1}=1$ and
\begin{equation}
\label{Eq:EspTheta-bis}
\Theta_p=\{T_p=(t_1,\ldots,t_p)\in [0,1]^p:\,  t_i\leq t_{i+1},\ \mbox{for } i=0,\ldots,p\}
\end{equation}
This could lead to ``degenerate'' options with $t_i=t_{i+1}$ for some values of $i$. However, the theory we develop below using \eqref{Eq:EspTheta} could be carried out alternatively with \eqref{Eq:EspTheta-bis} as long as we adopt the criterion of reducing the dimension of those vectors $(t_1,\ldots,t_d)$ with ties in the coordinate values by keeping just one coordinate for each different value.  
In this way, for example, $(1/3,1/3,1/2,3/4)$ would be interpreted just as $(1/3,1/2,3/4)$.

\

\noindent \textit{Two additional, equivalent optimality criteria}

A second optimality criterion, equivalent  to that based on $Q_1$, arises if we take into account that, 
from the reproducing property, when the regression function  is a finite linear combination of the form $\sum_{i=1}^p \beta_i K(t_i,\cdot)$, model \eqref{Eq:RKHSmodel} reduces to the usual finite dimensional multiple regression model:
\begin{equation}\label{Eq:RKHSmodel-finito}
Y = \sum_{i=1}^p \beta_j (X(t_j)-m(t_j)) + \varepsilon.	
\end{equation}	

Then, another sensible approach for variable selection is to choose those points $t_1,\ldots,t_p$ giving the best approximation of the true regression function $\beta$ in terms of a finite linear combination of the form $\sum_{i=1}^p \beta_i K(t_i,\cdot)$.  It is quite natural to use the norm in $\mathcal{H}(K)$ to assess this approximation since both $\beta$ and these finite linear combinations live in this RKHS. This approach amounts to  find the vector $T_p\in \Theta_p$ that minimizes the function
\begin{equation}
\label{Eq:Q2}
Q_2(T_p) := \min_{(\beta_1,\ldots,\beta_p)\in\mathbb{R}^p}\|\beta - \sum_{j=1}^p \beta_j K(t_j,\cdot)\|^2_K.
\end{equation}

Proposition \ref{Prop:Equivalence} below shows that the variable selection procedures defined by \eqref{Eq:Q1} and \eqref{Eq:Q2}, although apparently different, are indeed equivalent. Moreover, in the proof of Proposition \ref{Prop:Equivalence} 
we will see that the minimum in the expressions of $Q_1$ and $Q_2$ is achieved at the value  
$$
(\beta_1^*,\ldots,\beta_p^*)=\Sigma_T^{-1}c_T, 
$$ 
where $c_T=(\mbox{cov}(X(t_1),Y),\ldots, \mbox{cov}(X(t_p),Y))'$ and $\Sigma_T$ is the covariance matrix of $X_T$, for $T=(t_1,\ldots,t_p)$.

In addition, we show that the $Q_1$ and $Q_2$-based criteria are also both equivalent to a third criterion, defined in terms of a functional $Q_0$,  which only depends on the covariances $K(t_i,t_j)$ and $\mbox{Cov}(X(t_i),Y)$ for $i,j=1,\ldots,p$. This $Q_0$ criterion turns out to be especially useful to implement the method in practice.

\begin{proposition}
\label{Prop:Equivalence}
Assume that $Y$ and $X$ fulfil the RKHS functional regression model \eqref{Eq:RKHSmodel}. Then,
\begin{equation}
\label{Eq:Equivalence}
\argmin_{T_p\in\Theta_p} Q_1(T_p)=\argmin_{T_p\in\Theta_p} Q_2(T_p) = \argmax_{T_p\in\Theta_p} Q_0(T_p),
\end{equation}
where $Q_1$ and $Q_2$ are defined in \eqref{Eq:Q1} and \eqref{Eq:Q2} respectively, and 
\begin{equation}
\label{Eq:Q0}
Q_0(T_p):=c'_{T_p}\Sigma^{-1}_{T_p}c_{T_p},
\end{equation}
with $c_{T_p} = (\mbox{Cov}(X(t_1), Y),\ldots,\mbox{Cov}(X(t_p), Y))'$ and $\Sigma_{T_p}$ the $p \times p$ matrix with entries $K(t_i,t_j)$.
\end{proposition}

\begin{proof}
Since  $\mathbb{E}(\varepsilon)=0$ and  $\langle X(t), \epsilon\rangle_2 = 0$,
\[
\|Y - \sum_{j=1}^p \alpha_j (X(t_j) - m(t_j))\|^2_2 =  \|\langle X,\beta\rangle_K - \sum_{j=1}^p \alpha_j (X(t_j)-m(t_j))\|^2_2 + \sigma^2.
\]
On the other hand, Loève's isometry implies
\[
 \|\langle X,\beta\rangle_K - \sum_{j=1}^p \alpha_j (X(t_j)-m(t_j))\|^2_2 = \|\beta - \sum_{j=1}^p \alpha_j K(t_j,\cdot)\|^2_K.
 \]
From the last two equations, it follows that $Q_1(T_p)=Q_2(T_p)+\sigma^2$ and hence the first equality in \eqref{Eq:Equivalence}.

By the reproducing property,
\begin{equation}\label{Eq:Q2des}
\|\beta - \sum_{j=1}^p \alpha_j K(t_j,\cdot)\|^2_{K} = \|\beta\|^2_K - 2\sum_{j=1}^p \alpha_j\beta(t_j) + \sum_{i=1}^p\sum_{j=1}^p \alpha_i\alpha_j K(t_i,t_j).
\end{equation}

The function $K$ is positive semidefinite so that the last function is convex in $\alpha=(\alpha_1,\ldots,\alpha_p)$. By computing its gradient (with respect to $\alpha$) it is very easy to see that the minimum is achieved at $\alpha^* = (\alpha_1^*,\ldots, \alpha_p^*) = \Sigma^{-1}_{T_p}\beta_{T_p}$. Then,
\begin{equation}
\label{eq.auxequivalence}
Q_2(T_p) =  \|\beta - \sum_{j=1}^p \alpha^*_j K(t_j,\cdot)\|^2_K =  \|\beta\|^2_K - \beta'_{T_p}\Sigma^{-1}_{T_p}\beta_{T_p},
\end{equation}
where $\beta_{T_p}=(\beta(t_1),\ldots,\beta(t_p))'$.
Finally, using $\mathbb{E}(X(t)\varepsilon)=0$ and Equation \eqref{isometry} we get
\[
\mbox{Cov}(Y,X(t)) = \mathbb{E}\big[\langle X,\beta\rangle_K (X(t)-m(t))\big] = \Psi_X(\langle X,\beta\rangle_K)(t) = \beta(t).
\]
To obtain the last equality, recall that $\langle X,\beta\rangle_K=\Psi^{-1}_X(\beta)$. 
Therefore $\beta_{T_p}=c_{T_p}$ and, by \eqref{eq.auxequivalence}, $Q_2(T_p)=\|\beta\|^2_K - Q_0(T_p)$. This implies  the second equality in \eqref{Eq:Equivalence}.
\end{proof}

The criterion provided by $Q_0$ (or $Q_1, Q_2$) for variable selection was already considered by \cite{mckeague2010}, for $p=1$, when $X(t)$ is a fractional Brownian Motion with Hurst exponent $H\in (0,1)$, and by  \cite{muller2016} for $p\geq 1$ in the usual $L^2$ functional regression model. The RKHS formalism we incorporate here provides a simple way to describe the scenario under which variable election would lead to the optimal solution (with no loss of information). Variable selection is specially suitable when the true regression model is sparse, meaning that the response depend on the explanatory variables through their values at a finite small number of $p^*$ points. As it was mentioned before, this is the case under \eqref{Eq:RKHSmodel} when 
\begin{equation}
\label{Eq:Sparse}
\beta(t)=\sum_{j=1}^{p^*} \beta_j K(t_j^*,t).
\end{equation}
Let $T^*_{p^*} = (t_1^*,\ldots, t^*_{p^*})\in \Theta_{p^*}$. Then, it is clear that, under \eqref{Eq:Sparse},
\[
Q_2(T^*_{p^*}) = 0\leq Q_2(T_{p^*}),\ \ \mbox{for all}\ T_{p^*} \in \Theta_{p^*}.
\]
As a consequence, the true set of relevant variables $T^*_{p^*}$ is the one selected by the optimization of the functions in Proposition \ref{Prop:Equivalence}. In this reasoning we have considered the case when we know the actual number of points $p^*$ to be selected. In practice, this is not usually the case. However, notice that if we make a conservative choice, taking a number of variables $p$ larger than the true one ($p > p^*$), the true relevant variables $T^*_{p^*}$ will always be included among the selected ones. This is a consequence of the fact that, by \eqref{Eq:RKHSmodel}, \eqref{Eq:Sparse} and $\varepsilon \in \overline{\mathcal{L}}_X^\perp$, the projection of $Y$ onto $ \overline{\mathcal{L}}_X$ is the same as the projection onto the span of $X(t^*_1),\ldots, X(t^*_{p^*})$. This implies that, if we are looking for the subspace of dimension  $p>p^*$ that minimizes the distance of $Y$ to its projection (see Equation \eqref{Eq:Q1}), this space will always include the span of $X(t^*_1),\ldots, X(t^*_{p^*})$. 

Recall that we use the notation $T^* \prec T \in \Theta_p$ meaning that $T^*$ is a sub-vector of $T$, that is, that the components of $T^*$ are included within those of $T$. With this notation, what we have shown is that, under \eqref{Eq:Sparse},  $T^* = \left( t_1^*,\ldots,t_{p^*}^*\right) \prec \argmax  Q_0(T_p)$, for $p\geq p^*$. In Section \ref{Sec:NumberP} it is addressed the problem of what to do when $p^*$ is unknown from an empirical point of view.

%%%%%%%%%%%%%%%%%%%%%%%%%%%%%%%%%%%%%%%%%%%%%%%%%%%%%%%%%%%%%%%%%%%%%%%
%%%%%%%%%%%%%% 			Alternative expressions 		%%%%%%%%%%%%%%%%%%
%%%%%%%%%%%%%%%%%%%%%%%%%%%%%%%%%%%%%%%%%%%%%%%%%%%%%%%%%%%%%%%%%%%%%%%

%\subsection{Alternative interpretations of the optimization function}\label{Sec:OtrasInterp}

\subsection{A recursive expression}

The function $Q_0$ defined in \eqref{Eq:Q0} can be   rewritten in an alternative way, which is useful to analyze the gain when we add a new variable to a set of variables already selected. Moreover, this alternative expression paves the way for a sequential implementation of the variable selection method. To give the formula, besides the notation $c_T$ and $\Sigma_T$, introduced earlier, we will also use $c_j$ to denote $\mbox{cov}(X(t_j), Y)$, $\sigma_j^2$ to denote $\mbox{var} (X(t_j))$, and $c_{T_p,j}$ to denote the vector $\left( \mbox{cov}(X(t_1),X(t_j)), \ldots, \mbox{cov}(X(t_p),X(t_j)) \right)'$.

\begin{proposition}\label{Prop:IterativeQ0Matrices}
Given $T_{p+1}= (t_1,\ldots,t_p,t_{p+1})\in\Theta_{p+1}$, $p\geq 1$, and $T_p \prec T_{p+1}$, for some $p\geq 1$ such that the covariance matrices of the process $\Sigma_{T_{p+1}}$ are invertible for all $T_{p+1}\in \Theta_{p+1}$,
\begin{equation}\label{Eq:IterativeQ0Matrices}
Q_0(T_{p+1}) \hspace*{2mm}=\hspace*{2mm} Q_0(T_p) \hspace*{1mm}+\hspace*{1mm} \frac{\left( c_{T_p}'\Sigma_{T_p}^{-1}c_{T_p,p+1} - c_{p+1} \right)^2}{\sigma_{p+1}^2 - c_{T_p,p+1}'\Sigma_{T_p}^{-1}c_{T_p,p+1}}.
\end{equation}
\end{proposition}

 Equation \eqref{Eq:IterativeQ0Matrices} is useful to simplify other derivations in the paper and it is shown in Section \ref{Sec:ProofIterativeQ0Matrices}. Actually, it can be also proved that the quotient in \eqref{Eq:IterativeQ0Matrices} tends to zero when $t_{p+1}$ tends to one of the points in $T_p$, so that selecting a point too close to one of those already selected is redundant and non-informative according to this criterion. The proof of this fact is  quite technical so it is omitted.
 
Equation \eqref{Eq:IterativeQ0Matrices} already appears in the well-known forward selection method for variable selection in multiple regression  (see e.g. \cite{miller2002}, Section 3.2).  A modification of the resulting expression is also used in the variable selection method proposed by \cite{yenigun2015}, still in the multivariate regression setting. In such alternative version,  the usual covariance is replaced by the distance covariance, defined in \cite{szekely2013}.

The quotient in Equation \eqref{Eq:IterativeQ0Matrices} can be written in a more insightful way, as shown in the following result.
\begin{proposition}\label{Prop:IterativeQ0Vars}
In the above defined setup, denoting $X(T_p) = \big( X(t_1),\ldots, X(t_p) \big)'$, and $Y_{T_p} = P_{\mbox{span}\{X(t_i),t_i\in T_p\}} Y$ and $X_{T_p}(t_{p+1}) = P_{\mbox{span}\{X(t_i),t_i\in T_p\}} X(t_{p+1})$ the projections on $\mbox{span}\{X(t_i),t_i\in T_p\}$, 
\begin{equation}\label{Eq:IterativeQ0}
Q_0(T_{p+1}) \hspace*{2mm}=\hspace*{2mm} Q_0(T_p) \hspace*{1mm}+\hspace*{1mm} \frac{cov^2\big(Y - Y_{T_p},\hspace*{1mm} X(t_{p+1})\big)}{\mbox{var}\big( X(t_{p+1}) - X_{T_p}(t_{p+1})\big)}.
\end{equation}
\end{proposition}

%The proof of this result can be found in Section \ref{Sec:ProofIterativeQ0Vars}. 

The quotient of Equation \eqref{Eq:IterativeQ0} is known as \textit{part correlation coefficient} or \textit{semi-partial correlation coefficient}, a quantity which appears in several techniques dealing with multivariate data. The proof of this result can be found in Subsection \ref{Sec:ProofIterativeQ0Vars}.

\section{Sample properties of the variable selection method}
\label{Sec:SampleProperties}

\subsection{The proposed method}\label{Subsec:themethod}

In order to carry out the variable selection in practice, we have to estimate the function $Q_0$ from a sample $(Y_1,X_1), \ldots, (Y_n, X_n)$ of independent observations drawn from the model \eqref{Eq:RKHSmodel}. The most natural estimator is given by $\widehat{Q}_n(T_p) = \widehat{c}_{T_{p}}'\widehat{\Sigma}_{T_{p}}^{-1}\widehat{c}_{T_{p}}$, where $\widehat{c}_{T_{p}}$ and $\widehat{\Sigma}_{T_{p}}$ are the sample versions of $c_{T_p}$ and $\Sigma_{T_p}$, respectively, based on the sample mean $\overline{X}(t)= n^{-1}\sum_{i=1}^n X_i(t)$ and the sample covariances 
\[
\widehat{\mbox{Cov}}(X(s),X(t)) = n^{-1}\sum_{i=1}^n X_i(s)X_i(t) - \overline{X}(s)\overline{X}(t)
\]
of the trajectories. Then, if we want to select $p$ variables we propose to use $\widehat{T}_{p,n}$, where
\begin{equation}
	\label{Eq:EstimT}
	\widehat{T}_{p,n} :=  \argmax_{T_{p}\in \Theta_{p}} \, \widehat{Q}_n(T_{p}) =   \argmax_{T_{p}\in \Theta_{p}} \, \widehat{c}_{T_{p}}'\widehat{\Sigma}_{T_{p}}^{-1}\widehat{c}_{T_{p}}.
\end{equation}

\subsection{Asymptotic results}\label{Subsec:asymptotics}

In the following results, we will analyze the asymptotic behavior of this proposal.  We start with three preliminary results that may be of some interest by themselves. First we prove that, under some moment conditions,  the sample mean and  covariance functions of $X$ converge uniformly  a.s. to their population counterparts:

\begin{lemma} 
	\label{Lemma:uniformidad}
	Assume that the process $X$ has continuous trajectories with continuous mean and covariance functions and that it fulfils that  $\mathbb{E}\big[\sup_{t\in [\delta,1]} X(t)^2\big]<\infty $, for a certain $\delta\geq 0$. Then,
	\begin{equation}
		\label{eq:uniformidad-covarianza}
		\sup_{s, t\in [\delta,1]} |\widehat{\mbox{Cov}}(X(s),X(t)) -  \mbox{Cov}(X(s),X(t))| \overset{a.s.}{\rightarrow} 0.
	\end{equation}
\end{lemma}

\begin{proof}
	Note that the assumption implies $\mathbb{E}\big[\sup_{t\in [\delta,1]} |X(t)|\big]<\infty $ and the stochastic process $\{X(t): t\in [\delta,1]\}$ has finite strong expectation with trajectories in $\mathcal{C}[\delta, 1]$, which is a separable Banach space. Then, we can apply Mourier's SLLN (see e.g. Theorem 4.5.2 of \cite{laha1979}, p. 452) to conclude
	\begin{equation}
	\label{eq:uniformidad-media}
	\sup_{t\in [\delta,1]} |\overline{X}(t) -  m(t)| \overset{a.s.}{\rightarrow} 0,
	\end{equation}
	 Similarly, the process  $Z(s,t):=X(s)X(t)$, with trajectories in $\mathcal{C}([\delta,1]^2)$, is such that its strong expectation exists. Indeed, since
	\[
	0\leq (|X(s)| - |X(t)|)^2 = |X(s)|^2 + |X(t)|^2 - 2 |Z(s,t)|,
	\]
	it holds
	\[
	\mathbb{E}\big(\sup_{s,t\in [\delta,1]} |Z(s,t)|\big) \leq \mathbb{E}\big(\sup_{t\in [\delta,1]}|X(t)|^2\big)  < \infty.
	\]
	Moreover, $\mathcal{C}([\delta,1]^2)$ is separable since $[\delta,1]^2$ is compact.  Then,  Mourier's SLLN  and \eqref{eq:uniformidad-media} imply \eqref{eq:uniformidad-covarianza}.
\end{proof}

Next, we prove  that both $\widehat{Q}_n$ and $Q_0$ are continuous functions for any $p\geq 1$:

\begin{lemma} 
	\label{Lemma:continuidad}
	Assume that the process $X(t)$ has continuous mean and covariance functions. Let $p\geq 1$ and $\Theta_p=\Theta_p(\delta)$ be such that the assumptions of Lemma \ref{Lemma:uniformidad} hold. In addition, asssume that the covariance matrix $\Sigma_{T_p}$ is invertible for all $T_p\in \Theta_p$. Then, the functions $\widehat{Q}_n$ and $Q_0$ are continuous on $\Theta_p$.
\end{lemma}

\begin{proof}
	Fix $p\geq 1$. First, we prove that $Q_0$ is continuous. Since  the process $X(t)$ has continuous mean and covariance functions we have that 
	\[
	c_{T_p}  = (\mbox{Cov}(X(t_1), Y),\ldots,\mbox{Cov}(X(t_p), Y))'
	\]
	is continuous on $\Theta_p$. On the other hand, since the entries of $\Sigma_{T_p}$ are continuous on $[0,1]^2$, $\mbox{det}(\Sigma_{T_p})$ is also continuous on $\Theta_p$, where $\mbox{det}(\Sigma)$ stands for the determinant of $\Sigma$.  By assumption, $\mbox{det}(\Sigma_{T_p})>0$ for all $T_p\in\Theta_p$. Since $\Theta_p$ is compact, $
	\inf_{T_p\in\Theta_p} \mbox{det}(\Sigma_{T_p})>0$. Observe that
	\[
	\Sigma_{T_p}^{-1} = \frac{\mbox{adj}(\Sigma_{T_p})}{\mbox{det}(\Sigma_{T_p})},
	\]
	where $\mbox{adj}(\Sigma)$ denotes the adjugate of $\Sigma$. As a consequence, the entries of $\Sigma_{T_p}^{-1}$ are continuous on $\Theta_p$, and hence the function $Q_0$ is also continuous.
	
	The proof for $\widehat{Q}_n$ is analogous with the only difference that in this case we must ensure that $\inf_{T_p\in\Theta_p} \mbox{det}(\widehat{\Sigma}_{T_p})>0$ with probability 1. Notice that  for all $n\geq 1$ and $T_p\in \Theta_p$, $\mbox{det}(\widehat{\Sigma}_{T_p})>0$ a.s. On the other hand, from  \eqref{eq:uniformidad-covarianza} it follows that
	\begin{equation}
		\label{eq:determinant}
		\sup_{T_p\in \Theta_p}| \mbox{det}(\widehat{\Sigma}_{T_p}) - \mbox{det}(\Sigma_{T_p}) |\overset{a.s}{\rightarrow} 0.
	\end{equation}
	We have seen before that $\inf_{T_p\in\Theta_p} \mbox{det}(\Sigma_{T_p})>0$. Then, with probability 1, there exists $n_0$ such that if $n\geq n_0$, $\inf_{T_p\in\Theta_p} \mbox{det}(\widehat{\Sigma}_{T_p})>0$. 
\end{proof}

The two previous lemmas allow us to prove the uniform convergence on $\Theta_p$ of the empirical criterion for variable selection to the theoretical one.

\begin{lemma}
	\label{Lemma:covergenciauniforme}
	Under the assumptions of Lemma \ref{Lemma:continuidad}, it holds that 
	\[
	\sup_{T_p\in\Theta_p} |\widehat{Q}_n(T_{p}) - Q_0(T_p)| \overset{a.s.}{\rightarrow} 0.
	\]
\end{lemma}

\begin{proof}
	It is enough to establish the uniform convergence a.s. of the coordinates of $\widehat{c}_{T_p}$ and the entries of $\widehat{\Sigma}^{-1}_{T_p}$ to those of $c_{T_p}$ and $\Sigma^{-1}_{T_p}$ respectively. 
	
	Equation \eqref{eq:uniformidad-media} and the same argument leading to \eqref{eq:uniformidad-media} but applied to the process $Z(t)=X(t)Y$ yield
	\begin{equation}
		\label{Eq:uniformidad-covarianzaXY}
		\sup_{t\in[\delta,1]} \left|\frac{1}{n}\sum_{i=1}^n (X_i(t)-\overline{X}(t))(Y_i-\overline{Y}) - \mbox{Cov}(X(t),Y)\right| \overset{a.s}\rightarrow 0,
	\end{equation}
	and hence the uniform convergence a.s. of the coordinates of $\widehat{c}_{T_p}$  to those of $c_{T_p}$.
	
	Finally, observe that 
	$
	\widehat{\Sigma}_{T_p}^{-1} = \frac{\mbox{adj}(\widehat{\Sigma}_{T_p})}{\mbox{det}(\widehat{\Sigma}_{T_p})}.
	$
	Then, from \eqref{eq:uniformidad-covarianza}, \eqref{eq:determinant} and $\inf_{T_p\in\Theta_p} \mbox{det}(\Sigma_{T_p})>0$ we conclude the uniform convergence a.s. of the  entries of $\widehat{\Sigma}^{-1}_{T_p}$ to those of $\Sigma^{-1}_{T_p}$.
\end{proof}

Now assume that the sparsity condition \eqref{Eq:Sparse} holds. Then,  $T^*_{p^*}=(t^*_1,\ldots,t^*_{p^*})\in\Theta_{p^*}$ is ``sufficient'' in the sense that the response  only depends on the regressor variable through the values $X(t_1^*),\ldots,X(t^*_{p*})$. We have already seen that $T^*_{p^*}$  is a global maximum of $Q_0$ (see the remark below Equation \eqref{Eq:Sparse}). In fact, we are going to prove that under mild conditions it is the only global maximum of $Q_0$ on $\Theta_{p^*}$ and that  the estimator $\widehat{T}_{p^*,n}$ (defined in \eqref{Eq:EstimT} with $p=p^*$) converges a.s. to $T^*_{p^*}$. Then, our proposal is able to identify consistently the true relevant points.

\begin{theorem}
	\label{Theorem:ConsistencyOfT}
	Assume \eqref{Eq:Sparse} holds and that the process $X(t)$ has continuous mean and covariance functions. Suppose also that  the assumptions of Lemma \ref{Lemma:uniformidad} hold for $p=p^*$, the covariance matrix $\Sigma_{T_{p^*}}$ is invertible for all $T_{p^*}\in \Theta_{p^*}$ and the covariance matrix $\Sigma_{T_{p^*}\cup S_{p^*}}$ is invertible for all $T_{p^*}, S_{p^*}\in \Theta_{p^*}$, with $T_{p^*}\neq S_{p^*}$. Then,
	\begin{enumerate}
		\item[(a)]   The point $T^*_{p^*}\in \Theta_{p^*}$, given by \eqref{Eq:Sparse}, is the only global maximum of $Q_0$ on $\Theta_{p^*}$.
		\item[(b)] If $\widehat{T}_{p^*,n}=\argmax_{T_{p^*}\in \Theta_{p^*}} \, \widehat{Q}_n(T_{p^*})$, then $\widehat{T}_{p^*,n} \to T^*_{p^*}$ a.s. as $n\to\infty$.
		\item[(c)] $\widehat{T}_{p^*,n}$ converges to $T^*_{p^*}$ in quadratic mean, that is, $\mathbb{E}\big(\|\widehat{T}_{p^*,n}- T^*_{p^*}\|^2_2\big) \to 0$, as $n\to\infty$.
	\end{enumerate}
\end{theorem}

\begin{proof} (a)  In view of \eqref{Eq:Equivalence}, it is enough to prove that $T^*:= T^*_{p^*}$ is the unique global minimum of 
	\[
	Q_1(T_{p^*})=\Vert Y - Y_{T_{p^*}}\Vert_2^2=\Vert Y_{T^*} - Y_{T_{p^*}}\Vert_2^2 + \mbox{Var}(\varepsilon).
	\]  
	The expression above readily shows that $T^*$ minimizes $Q_1$. Suppose that there exists another minimum $S^* \in \Theta_{p^*}$ such that $S^*\neq T^*$. Then, we must have $\Vert Y_{T^*} - Y_{S^*}\Vert_2^2 =0$ and hence $Y_{T^*} - Y_{S^*}=0$ a.s. As a consequence, using the notation $\widetilde{X}(t) = X(t) - m(t)$, there  exist coefficients $\beta_j$ and $\alpha_j$ such that $\sum_{j=1}^{p^*} \beta_j \widetilde{X}(t_j^*) - \sum_{j=1}^{p^*} \alpha_j \widetilde{X}(s_j^*)=0$ a.s.,  for  $T^*, S^* \in \Theta_{p^*}$ with $S^*\neq T^*$. This fact contradicts the assumption that the  covariance matrix $\Sigma_{T^*\cup S^*}$ must be invertible. Therefore, $T^*=S^*$.
	
	(b) Since the functions $\widehat{Q}_n$ and $Q_0$ are continuous on $\Theta_{p^*}$ (by Lemma \ref{Lemma:continuidad}) and the sequence of functions $\widehat{Q}_n$ tends uniformly a.s. to $Q_0$ on $\Theta_{p^*}$ (by Lemma \ref{Lemma:covergenciauniforme}) the fact that $Q_0$ has a unique maximum on $\Theta_{p^*}$(part (a)) implies that $\widehat{T}_{p^*,n}$  converges almost surely to $T^*_{p^*}$. 
	
	(c) From part (b), we have $\|\widehat{T}_{p^*,n}- T^*_{p^*}\|_2\to 0$ a.s. as $n\to\infty$. Moreover, since both $\widehat{T}_{p^*,n}$ and $T^*_{p^*}$ belong to $\Theta_{p^*}$,
	\[
	\|\widehat{T}_{p^*,n}- T^*_{p^*}\|_2\leq \|\widehat{T}_{p^*,n}\|_2 + \|T^*_{p^*}\|_2 \leq 2p^*.
	\]
	The result follows from  dominated convergence theorem (using $2p^*$ as the integrable dominating function).
\end{proof}

Once we have selected $p^*$ points, we can use them to predict the response variable. The optimal predictions (in a square mean sense) are given by:
\[
\widehat{Y}_{\widehat{T}_{p^*}} = \widehat{\beta}_1 \widetilde{X}(\widehat{t}_1) + \cdots + \widehat{\beta}_{p^*} \widetilde{X}(\widehat{t}_{p^*}),
\]
where $\widetilde{X}(t) = X(t) -m(t)$, and $\widehat{\beta}_{\widehat{T}_{p^*}} := (\widehat{\beta}_1,\ldots,\widehat{\beta}_{p^*})' = \widehat{\Sigma}_{\widehat{T}_{p^*}}^{-1}\widehat{c}_{\widehat{T}_{p^*}}$. On the other hand, the prediction we would use under condition \eqref{Eq:Sparse}
if we knew the true relevant points and the true values of the parameters of the model would be
\[
Y_{T^*} = \beta^*_1 \widetilde{X}(t^*_1) + \cdots + \beta^*_{p^*} \widetilde{X}(t^*_{p^*}),
\]
where now $\beta^*_{T^*_{p^*}}:=(\beta^*_1,\ldots,\beta^*_{p^*})'= \Sigma^{-1}_{T^*_{p^*}} c_{T^*_{p^*}}$.
The following result refers to the asymptotic behavior of the data-driven predictions $\widehat{Y}_{\widehat{T}_{p^*}}$. It is shown that they converge a.s. and in quadratic mean to the oracle values $Y_{T^*}$.

\begin{theorem}
	\label{Theorem:ConsistencyOfY}
	Under the assumptions of Theorem \ref{Theorem:ConsistencyOfT}, $\widehat{Y}_{\widehat{T}_{p^*}} \overset{a.s.}{\rightarrow} Y_{T^*}$. If, in addition, there exists $\eta>0$ such that $\mathbb{E}\big[\sup_{t\in [\delta,1]} |X(t)|^{2+\eta}\big]<\infty $ then  $\widehat{Y}_{\widehat{T}_{p^*}} \overset{L^2}{\longrightarrow} Y_{T^*}$, as $n\to\infty$.
\end{theorem}

\begin{proof}
	For simplicity, denote $\widehat{T}:=\widehat{T}_{p^*}$ and $T^*:=T^*_{p^*}$. Observe that
	\begin{align*}
		|\widehat{Y}_{\widehat{T}} - Y_{T^*} | =& |\widehat{c}'_{\widehat{T}}\widehat{\Sigma}_{\widehat{T}}^{-1}\widetilde{X}(\widehat{T}) - c'_{T^*} \Sigma_{T^*}^{-1} \widetilde{X}(T^*)| \\ 
		\leq & 
		|\widehat{c}'_{\widehat{T}}\widehat{\Sigma}_{\widehat{T}}^{-1}\widetilde{X}(\widehat{T}) -
		c'_{\widehat{T}}\Sigma_{\widehat{T}}^{-1}\widetilde{X}(\widehat{T})| + |c'_{\widehat{T}}\Sigma_{\widehat{T}}^{-1}\widetilde{X}(\widehat{T})-c'_{T^*} \Sigma_{T^*}^{-1} \widetilde{X}(T^*)|\\
		\leq &
		\|\widehat{c}'_{\widehat{T}}\widehat{\Sigma}_{\widehat{T}}^{-1} -  c'_{\widehat{T}}\Sigma_{\widehat{T}}^{-1}\|_2 \, \| \widetilde{X}(\widehat{T})\|_2 + |c'_{\widehat{T}}\Sigma_{\widehat{T}}^{-1}\widetilde{X}(\widehat{T})-c'_{T^*} \Sigma_{T^*}^{-1} \widetilde{X}(T^*)|\\
		\leq &
		\sup_{T\in \Theta_{p^*}} \|\widehat{c}'_T\widehat{\Sigma}_T^{-1} -  c'_T\Sigma_T^{-1}\|_2 \sup_{T\in \Theta_{p^*}} \| \widetilde{X}(T)\|_2\\
		+ & |c'_{\widehat{T}}\Sigma_{\widehat{T}}^{-1}\widetilde{X}(\widehat{T})-c'_{T^*} \Sigma_{T^*}^{-1} \widetilde{X}(T^*)|.
	\end{align*}
	Then, to prove $\widehat{Y}_{\widehat{T}}\to Y_{T^*}$ a.s. it is enough to see that the two addends of the last expression go to 0 a.s. Observe that $\sup_{T\in \Theta_{p^*}} \|\widehat{c}'_T\widehat{\Sigma}_T^{-1} -  c'_T\Sigma_T^{-1}\|_2 \to 0$ a.s., as $n\to\infty$, by \eqref{eq:uniformidad-covarianza} and \eqref{Eq:uniformidad-covarianzaXY}. Moreover, since $X(t)$ has continuous trajectories and continuous mean function, and $\Theta_{p^*}$ is compact, we have $\sup_{T\in \Theta_{p^*}} \| \widetilde{X}(T)\|_2<\infty$. Finally, the continuity of $c_T$, $\Sigma_T$ and $\widetilde{X}(T)$, together with Theorem \ref{Theorem:ConsistencyOfT}(b), imply that $|c'_{\widehat{T}}\Sigma_{\widehat{T}}^{-1}\widetilde{X}(\widehat{T})-c'_{T^*} \Sigma_{T^*}^{-1} \widetilde{X}(T^*)|\to 0$ a.s., as $n\to\infty$.
	
	In order to prove $\widehat{Y}_{\widehat{T}} \overset{L^2}{\longrightarrow} Y_{T^*}$, as $n\to\infty$, we will check that there exists $\eta>0$ such that $\sup_{n}\mathbb{E}\vert Y_{\widehat{T}}\vert^{2+\eta} < \infty$, which in turn implies that the sequence $\widehat{Y}^2_{\widehat{T}}$ is uniformly integrable. Then, we can apply that a uniformly integrable sequence of random variables which converges in probability, also converges in $L^1$ (see e.g. Proposition 6.3.2 and the corollary of Proposition 6.3.3 in \cite{laha1979}).
	
	By assumption, there exists $\eta>0$  such that $\mathbb{E} \big( \sup_{t\in [\delta,1]} |\widetilde{X}(t)|^{2+\eta}\big)<\infty$. Observe that
	\[
	\vert \widehat{Y}_{\widehat{T}}\vert^{2+\eta} = | \widehat{c}'_{\widehat{T}}\widehat{\Sigma}_{\widehat{T}}^{-1}\widetilde{X}(\widehat{T})|^{2+\eta} \leq \|  \widehat{c}'_{\widehat{T}}\widehat{\Sigma}_{\widehat{T}}^{-1} \|^{2+\eta}_2 \, \| \widetilde{X}(\widehat{T})\|^{2+\eta}_2.
	\]
	We have seen that $\sup_{T\in \Theta_{p^*}} \|\widehat{c}'_T\widehat{\Sigma}_T^{-1} -  c'_T\Sigma_T^{-1}\|_2 \to 0$ a.s., as $n\to\infty$. Then, given $\varepsilon>0$, for large enough $n$,
	\[
	\|  \widehat{c}'_{\widehat{T}}\widehat{\Sigma}_{\widehat{T}}^{-1} \|^{2+\eta}_2\leq \varepsilon + \sup_{T\in\Theta_{p^*}} \|c'_T\Sigma_T^{-1}\|_2^{2+\eta} := C<\infty.
	\]
	From the last two displayed equations, for large enough $n$,
	\[
	\mathbb{E}\vert \widehat{Y}_{\widehat{T}}\vert^{2+\eta} \leq C\,  \mathbb{E}\|\widetilde{X}(\widehat{T})| \|_2^{2+\eta}\leq C' \, \mathbb{E} \big( \sup_{t\in [\delta,1]} |\widetilde{X}(t)|^{2+\eta}\big)<\infty,
	\]
	where $C' = C (p^*)^{(2+\eta)/2}$. Since the last upper bound does not depend on $n$, we get $\sup_{n}\mathbb{E}\vert Y_{\widehat{T}}\vert^{2+\eta} < \infty$.
\end{proof}

%%%%%%%%%%%%%%%%%%%%%%%%%%%%%%%%%%%%%%%%%%%%%%%%%%%%%%%%%%%%%%%%%%%%%%%
%%%%%%%%%%%%%% 			Practical issues				%%%%%%%%%%%%%%%%%%
%%%%%%%%%%%%%%%%%%%%%%%%%%%%%%%%%%%%%%%%%%%%%%%%%%%%%%%%%%%%%%%%%%%%%%%

\subsection{Sequential approach}

The number of combinations of variables is usually too large to carry out an exhaustive search to find the optimal $p^*$ variables, even for small values of $p^*$. Then, in practice we need to define a search strategy  to perform the selection. That is, we must decide how to explore the space of all possible combinations of variables.  We propose to use the sequential approach we describe in this section.

Observe that a proof analogous to that  of Equations \eqref{Eq:IterativeQ0Matrices} and \eqref{Eq:IterativeQ0} also gives their corresponding  sample versions: 
\begin{eqnarray}
	\widehat{Q}_n(T_{p+1}) &=& \widehat{Q}_n(T_p) \hspace*{1mm}+\hspace*{1mm} \frac{\left( \widehat{c}_{T_p}'\widehat{\Sigma}_{T_p}^{-1}\widehat{c}_{T_p,p+1} - \widehat{c}_{p+1} \right)^2}{\widehat{\sigma}_{p+1}^2 - \widehat{c}_{T_p,p+1}'\widehat{\Sigma}_{T_p}^{-1}\widehat{c}_{T_p,p+1}}, \label{Eq:IterativeQnMatrices}\\[1em]
	%%%
	\widehat{Q}_n(T_{p+1}) &=& \widehat{Q}_n(T_p) \hspace*{1mm}+\hspace*{1mm} \frac{\widehat{\mbox{cov}}^2\big(Y - \widehat{Y}_{T_p},\hspace*{1mm} X(t_{p+1})\big)}{\widehat{\mbox{var}}\big( X(t_{p+1}) - \widehat{X}_{T_p}(t_{p+1})\big)}.\label{Eq:IterativeQn}
\end{eqnarray}
These equations suggest a sequential way to carry out the variable selection. In each step, we  find the variable $t_{p+1}\in [\delta,1]$ maximizing the quotient in the equations above.  In this way, we obtain  nested subsets of variables, since $T_p \prec T_{p+1}$. This greedy method does not guarantee the convergence to   the global maximum of $\widehat{Q}_n$, but it shows a good behavior in practice, as we will show later on.

%Other important point is the inversion of the covariance matrices. Under the assumptions of Theorem \ref{Teo:ConsT}, there should be no problem, since the process should be no-degenerated. However, in practice, if the width of the grid is quite small, or if the variables of the process are very correlated, it could derive into some numerical errors. Nevertheless, the standard platforms as Matlab or R have inversion algorithms to deal with this situations.

%%%%%%%%%%%%%%%%%%%%%%%%%%%%%%%%%%%%%%%%%%%%%%%%%%%%%%%%%%%%%%%%%%%%%%%
%%%%%%%%%%%%%%%%%%%%%%%%%%%%%%%%%%%%%%%%%%%%%%%%%%%%%%%%%%%%%%%%%%%%%%%
%%%%%%%%%%%%%% 				Estimating p			   %%%%%%%%%%%%%%%%%%
%%%%%%%%%%%%%%%%%%%%%%%%%%%%%%%%%%%%%%%%%%%%%%%%%%%%%%%%%%%%%%%%%%%%%%%
%%%%%%%%%%%%%%%%%%%%%%%%%%%%%%%%%%%%%%%%%%%%%%%%%%%%%%%%%%%%%%%%%%%%%%%

\section{Estimating the number of variables}\label{Sec:NumberP}

As discussed above, our method for variable selection works, whenever the regression function $\beta$ belongs to the RKHS space associated with the covariance function $K$. It is based on the idea of asymptotically minimizing (on $T_p=(t_1,\ldots,t_p)$) the residuals
$$
Q_1(T_p) := \min_{(\beta_1,\ldots,\beta_p)\in\mathbb{R}^p} \|Y - \sum_{j=1}^p \beta_j \widetilde X(t_j)\|^2_2=\|Y - \sum_{j=1}^p \beta_j^* \widetilde X(t_j)\|^2_2,
$$
where $\widetilde X(t_j)=X(t_j)-m(t_j)$ and $(\beta_1^*,\ldots,\beta_d^*)'=\Sigma_T^{-1}c_T$, $\Sigma_T$ being the covariance matrix of $(X(t_1),\ldots,X(t_p))'$ and $c_T$ the vector whose $j$-th component is $\mbox{cov}(X(t_j),Y)$. As proved in Proposition \ref{Prop:Equivalence}, this 
amounts to asymptotically maximize  the function $Q_0$ defined in \eqref{Eq:Q0}, which in turn is equivalent to minimize the function $Q_2$, defined in \eqref{Eq:Q2}. Also, the functions, $Q_1$ and $Q_2$ agree up to an additive constant and both agree with $Q_0$ up to a change of sign plus an additive constant. 

Throughout this section we assume the validity of the sparsity assumption \eqref{Eq:RKHSmodel-finito}, that is, we assume that the slope function $\beta$ has the form $\beta=\sum_{j=1}^{p*}\beta_j K(t_j^*,\cdot)$, as stated in Equation \eqref{Eq:Sparse},
for some constants $\beta_1,\ldots,\beta_{p^*}\in {\mathbb R}$ and for $T_{p^*}^*=(t_1^*,\ldots,t_{p^*}^*)$. In this case, 
we can properly speak of a specific target set 
of   ``true'' variables $T^*=T_{p^*}^*=(t_1^*,\ldots,t_p^*)$ to be selected and, in particular, of a ``true'' number $p^*$  of variables to select. 

Keeping in mind these facts, the following comments provide some clues and motivation for the data-based selection of $p^*$. They will be formalized in the statement and proof of Lemma \ref{Lemma:Qp+1-Qnto0} below.

\begin{itemize}
	
	\item[(a)] On the one hand, any selection of type
	$T_p=(t_1,\ldots,t_p)$ with $p<p^*$ is clearly sub-optimal, since it would lack some relevant information, contributed by the variables in $T^*$ not in $T_p$.
	
	\item[(b)] Likewise, a choice $T_p$ ``by excess'' with $T^*\prec T_p$ would not provide any benefit. To see this
	note that, under \eqref{Eq:RKHSmodel-finito}, the minimum of $Q_2$ is obviously attained at $T^*_p$ and the value of $Q_2$ at such minimum is 0, which cannot be improved. 
	
	\item[(c)] As a consequence, the maximum value of $Q_2$ for points with $p^*+1$ coordinates is attained at some $T_{p^*+1}$ such that $T^*\prec T_{p^*+1}$
	(that is, $T^*$ is a sub-vector of $T_{p^*+1}$)
	but, in any case, $Q_0(T_{p^*+1})-Q_0(T^*)=0$.
	
	\item[(d)] Then, the optimal $p^*$ is such that the maximum value of $Q_0(T_{p})$ agrees with that of $Q_0(T_{p^*})$ for any $T_p$ such that $T_{p^*}\prec T_p$. Thus $p^*$ is in fact the ``elbow'' value  in the plot of $p\mapsto Q_0(T_p^*)$ from which on the increase of the maximum values of $Q_0$ stops.
		
\end{itemize}

The following lemma will set the theoretical basis of our procedure of estimation of $p^*$. As a consequence of this result, a procedure to estimate $p^*$ is proposed in the next subsection.

\begin{lemma}\label{Lemma:Qp+1-Qnto0}
	Let us consider the model \eqref{Eq:RKHSmodel} under the assumption that $\beta$ can be expressed as $\beta(t)=\sum_{j=1}^{p^*}\beta_jK(t_j^*,\cdot)$, where $p^*$ is the minimal integer for which such representation holds. Define $\widehat Q_n^{max}(p)=\max_{T_p\in\Theta_p}\widehat Q_n(T_p)$. Then, under the assumptions of Lemma \ref{Lemma:continuidad} we have
	
	(a) \begin{equation}\label{Eq:Qp+1-Qnto0}
	\widehat Q_n^{max}(p^*+1)-\widehat Q_n^{max}(p^*)\to 0, \mbox{a.s.},
	\end{equation}
	where $p^*$ stands for the ``true'' number of variables  in the sparse model \eqref{Eq:Sparse} 
	
	(b) For all $p<p^*$,
	\begin{equation}\label{Eq:Qp+1-Qnto>0}
	\lim_n \left(\widehat Q_n^{max}(p+1)-\widehat Q_n^{max}(p)\right) > 0, \mbox{a.s.},
	\end{equation}
\end{lemma}

\begin{proof}
	(a) Let us first prove
	\begin{equation}\label{Eq:difQ_n0}
	Q_0(T_{p^*+1}^*)-Q_0(T_{p^*}^*)=0
	\end{equation}		
	To see this, note that in the proof of Proposition \ref{Prop:Equivalence} we have proved $Q_0(T_p)=\|\beta\|^2_K - Q_2(T_p)$ with $Q_2(T_p) = \min_{(\beta_1,\ldots,\beta_p)\in\mathbb{R}^p}\|\beta - \sum_{j=1}^p \beta_j K(t_j,\cdot)\|^2_K$. Also, under \eqref{Eq:Sparse}, $Q_2(T_{p^*}^*)=0$ so that
	$Q_0(T_{p^*}^*)=\|\beta\|_K^2$, which is the maximum possible value of $Q_0$. 
	On the other hand, it is clear that 
	$Q_2(T_{p^*+1}^*)\leq Q_2(T_{p^*}^*)$ so that we must also have $Q_2(T_{p^*+1}^*)=0$ and $Q_0(T_{p^*+1}^*)=\|\beta\|_K^2$. This proves
	\eqref{Eq:difQ_n0}. Now conclusion \eqref{Eq:Qp+1-Qnto0} follows directly from the uniform convergence of $Q_n$ to $Q_0$, as established in Lemma \ref{Lemma:covergenciauniforme}. 
	
	(b) Similarly to part (a) we only need to prove
	\begin{equation}\label{Eq:difQ_n>0}
	Q_0(T_{p+1}^*)-Q_0(T_{p}^*)>0\ \mbox{for all } p<p^*.
	\end{equation}
	Indeed, assume we have 
	$Q_0(T_{p+1}^*)-Q_0(T_{p}^*)=0$ for some $p<p^*$.
	Then, since the prediction error $Q_1(T_p)$ defined in \eqref{Eq:Q1} satisfies
	$Q_1(T_p)=-Q_0(T_p)+\|\beta\|^2_K+\sigma^2$
	we would have that the prediction error $Q_1(T^*_p)$ obtained with $p$ variables in the sparse model $Y=\sum_{j=1}^q\beta_jX(t_j)+\varepsilon$, with $\mbox{var}(\varepsilon)=\sigma^2$ would be the same, for $q=p$ and $q=p+1$. Then, by recurrence, we get that the error would be in fact the same, irrespective of the number of $q$ explanatory variables in the range $p,p+1,\ldots,p^*$. Thus, the linear model $Y=\langle\beta,X\rangle_K+\varepsilon$ holds for a regression function of type $\beta=\sum_{j=1}^{p}\beta_jK(t_j^*,\cdot)$ with $p<p^*$. This is a contradiction with the assumption that $p^*$ is the minimal value for which such model holds. 
	 
	Now the result follows from \eqref{Eq:difQ_n>0} and the a.s. uniform convergence of $\widehat Q_n$  to $Q_0$.	
\end{proof}

%%%%%%%%%%%%%%%%%%%%%%%%%%%%%%%%%%%%%%%%%%%%%%%%%%%%%%%%%%%%%%%%%%%%%%%
%%%%%%%%%%%%%% 			Estimador y consistencia	   %%%%%%%%%%%%%%%%%%
%%%%%%%%%%%%%%%%%%%%%%%%%%%%%%%%%%%%%%%%%%%%%%%%%%%%%%%%%%%%%%%%%%%%%%%

\subsection{The estimator of $p^*$}
 
 The above discussion suggests the following method to estimate $p^*$:
 \begin{enumerate}
 	\item Define
 	\begin{equation}\label{Eq:Delta}	
 	\Delta=\min_{p<p^*}\hspace*{2mm} (Q_0(T_{p+1}^*)-Q_0(T_p^*))
 	\end{equation}
 	Assume we are able to fix a value $\epsilon>0$ such that $\epsilon<\Delta$.
 	\item Define
 	\begin{equation}\label{Eq:EstP}
 	\widehat{p} = \min \left\lbrace p: \widehat{Q}_n^{max}(p+1)-\widehat{Q}_n^{max}(p)<\epsilon\right\rbrace,
 	\end{equation}
 \end{enumerate}

\begin{theorem}\label{Th:hatp-to-p}
	Under the assumptions of Lemma \ref{Lemma:Qp+1-Qnto0} the estimator $\widehat p$ defined in \eqref{Eq:EstP} fulfils $\widehat p\to p^*$, almost surely.
	\begin{proof}
		This result is a direct consequence of Lemma \ref{Lemma:Qp+1-Qnto0}
	\end{proof}	
\end{theorem} 
 
In practice, the calculation of $\widehat p$ could be made using techniques inspired in the change point detection methodology in time series. Thus, 
we could interpret the collection of values $L_n(p)=\log\big(\widehat{Q}_n^{max}(p+1)-\widehat{Q}_n^{max}(p)\big)$ for $p=1,\ldots$ as a time series. Then, we could  apply the usual $k$-means clustering algorithm to these values, with $k=2$.  We would estimate $\widehat p$ as the minimum value of $p$ such that all the values $L_n(p)$ with $p\geq\widehat p$ belong to a different cluster than that of $L_n(1)$. However, different approaches could be used. For instance in \cite{delaigle2012}, where empirical methods to select both $p$ and $T_p$ in functional classification are given, the authors suggest to set $\varepsilon$ equal to $\rho \widehat{Q}_n^{max}(1)$ for a pre-determined small $\rho$.

\section{When $p^*$ is not estimated: the conservative oracle property}\label{Sec:pnot}
 
Under the sparseness assumption \eqref{Eq:Sparse}, where $p^*$ is unknown, another sensible approach for the choice of the number $p$ of selected variables is to take a conservative, large enough value of $p$. 

The basic idea of this section is easy to state: suppose that a ``conservative oracle'' gives us a value $p$ such that $p>p^*$. Accordingly, we perform our variable selection procedure for such value $p$. This yields $p$ variables $\widehat t_1,\ldots,\widehat t_p$.  Then, we can be sure that the ``true'' variables $t_1^*,\ldots,t_{p^*}$ are very close to the $p^*$ variables in $\{\widehat t_1,\ldots,\widehat t_p\}$.  

The next result formalizes this property.

\begin{theorem}\label{Th:conservative}
	Let us consider the model \eqref{Eq:RKHSmodel} under the assumption that $\beta$ can be expressed as $\beta(t)=\sum_{j=1}^{p^*}\beta_jK(t_j^*,\cdot)$, where $p^*$ is the minimal integer for which such representation holds. Let $\widehat t_1,\ldots, \widehat t_p$ the variables selected by the method \eqref{Eq:EstimT}, where $p$ is a given value larger than $p^*$. Then, for all $\epsilon>0$,
	\begin{equation}\label{Eq:oracle}
	{\mathbb P}\left(t_i^*\in \bigcup_{j=1}^p(\widehat t_j-\epsilon,\widehat t_j+\epsilon),\ i=1,\ldots,p^*\right)=1,\  \mbox{eventually, as } n\to\infty.
	\end{equation} 
	\begin{proof}
		Recall that the choice of the variables $T_p=(t_1,\ldots,t_p)$ is performed by asymptotically maximizing the function $Q_0(T_p)$, defined in \eqref{Eq:Q0}. More precisely, as $Q_0$ depends on unknown population quantities, we in fact maximize the estimator $\widehat{Q}_n$ defined in Subsection \ref{Subsec:themethod}. 
		
		Now, let us note that the maximum of $Q_0$ is not unique. Indeed, we assume that the ``minimal'' sparse representation of $\beta$ has the form $\beta(t)=\sum_{j=1}^{p^*}\beta_jK(t_j^*,\cdot)$ but, of course, if $p>p^*$, we may formally put $\beta(t)=\sum_{i=1}^p\beta_iK(s_i,\cdot)$
		as long as the ``true'' optimal points $t_i^*,\ldots,t_{p^*}$ are among the $s_i$'s and all the coefficients $\beta_i$ not matching with such $t_i^*$'s are null. On the other hand, from the uniqueness of the function $\beta$, all the maxima of $Q_0(T_p)$ must have an expression of this type. 
		
		Let $T^*_p$ be one of these maxima. From the above discussion it follows that the coordinates $s_i$ of $T^*_p$ whose corresponding coefficients $\beta_i$ are not null, must necessarily coincide with one of the points $t_i^*$ in the minimal representation $\beta(t)=\sum_{j=1}^{p^*}\beta_jK(t_j^*,\cdot)$. Let us denote by $T^{**}_p$ the subvector of $T^*_p$ corresponding to these values $t_i^*$.
		Let $\widehat T^*_{p,n}$ be a maximizer 
		of $\widehat{Q}_n(T_p)$ and denote by $\widehat T^{**}_{p,n}$ the subvector of $\widehat T^*_{p,n}$ corresponding to the same 
		coordinates as those of $T^{**}_p$. 
		
		According to Lemma \ref{Lemma:covergenciauniforme}, $Q_n(T_p)$ converges to $Q_0(T_p)$ uniformly a.s. in $T_p$. This entails the a.s. convergence of the subvectors $\widehat T^{**}_{p,n}$ to $T^{**}_p$, since (as indicated above) all the maxima of $Q_0$ must content the subvector $T^{**}_p$. 
		
		Now, conclusion \eqref{Eq:oracle} is a direct consequence from the definition of almost sure convergence. 
		
	\end{proof}
\end{theorem}

This result is  reminiscent of the \textit{Sure Screening Property} defined in \cite{fan2008}, which is used to quantify the efficiency of multivariate variable selection methods. But, obviously, property \eqref{Eq:oracle} is adapted to cope with the functional nature of the data and the fact that the values $t_i$ range on a continuous domain.

%%%%%%%%%%%%%%%%%%%%%%%%%%%%%%%%%%%%%%%%%%%%%%%%%%%%%%%%%%%%%%%%%%%%%%%
%%%%%%%%%%%%%%%%%%%%%%%%%%%%%%%%%%%%%%%%%%%%%%%%%%%%%%%%%%%%%%%%%%%%%%%
%%%%%%%%%%%%%% 				Experiments				   %%%%%%%%%%%%%%%%%%
%%%%%%%%%%%%%%%%%%%%%%%%%%%%%%%%%%%%%%%%%%%%%%%%%%%%%%%%%%%%%%%%%%%%%%%
%%%%%%%%%%%%%%%%%%%%%%%%%%%%%%%%%%%%%%%%%%%%%%%%%%%%%%%%%%%%%%%%%%%%%%%
 
\section{Experiments}\label{Sec:experiments}

The purpose of this section is to give some insights on the practical behaviour of our proposal for variable selection, both in simulations and real data examples. We are aware that the design of these experiments is largely discretionary, as the range of possible models for simulation is potentially unlimited (especially in the case of functional data models) and there is also a considerable amount of real data examples currently available in the FDA literature. Still, our choices have not been completely arbitrary. We have tried to follow some objective criteria. First, the theoretical models chosen for the simulations must obviously include some situations in which our crucial ``sparseness'' assumption $\beta=\sum_j\beta_jK(t_j,\cdot)$ is fulfilled. As discussed above, such models are quite natural if we are willing to use variable selection techniques. Also, it looks reasonable to include at least one model in which this assumption is not valid. Regarding the real data,  we have just chosen two examples used in the recent literature for the purpose of checking other variable selection methods in functional regression settings.  

In any case, we would like to emphasize that we make here no attempt to draw any definitive conclusion on the performance of our method when compared with others. In our view, no unique empirical study can lead to safe, objective conclusions in this regard: the only reliable verdict should be given by the users community, after some time of practice with real data problems. Our purposes here are far more unassuming; we just want to provide some hints suggesting that our proposal

(a) has a satisfactory performance in the ``sparse'' models for which it has been designed,

(b)  can be implemented in practice, with an affordable computational cost, 

(c) could be hopefully competitive under other theoretical models, far from the ideal assumption $\beta=\sum_j\beta_jK(t_j,\cdot)$,

(d) has also a satisfactory practical performance in a couple of real data examples commonly used in the literature of variable selection.

%%%%%%%%%%%%%%%%%%%%%%%%%%%%%%%%%%%%%%%%%%%%%%%%%%%%%%%%%%%%%%%%%%%%%%%
%%%%%%%%%%%%%% 				Simulated data			   %%%%%%%%%%%%%%%%%%
%%%%%%%%%%%%%%%%%%%%%%%%%%%%%%%%%%%%%%%%%%%%%%%%%%%%%%%%%%%%%%%%%%%%%%%

\subsection{Simulation experiments}

Keeping in mind the above general lines, we next define the simulation models under study. In our context a ``model'' is defined by three elements: a stochastic process (from which the functional data are generated), a regression equation, of type $Y=\langle X,\beta\rangle_K+\varepsilon$ (or, more generally, $Y=g(X)+\varepsilon$) and an error variable $\varepsilon$. In what follows, $\varepsilon$ has been chosen in all cases as $\varepsilon\sim N(0,\sigma)$ with $\sigma = 0.2$.

We have considered six processes, covering a broad range of different situations. 
\begin{enumerate}
	\item \textit{Standard Brownian Motion} (Bm) $\{B(t),\ t\in[0,1]\}$.
	
	\item \textit{Geometric Brownian Motion} (gBm). This non-Gaussian process is also known as exponential Brownian motion. It can be defined just by $X(t)=e^{B(t)}$. 
	
	\item \textit{Integrated Brownian Motion} (iBm): it is obtained as $X(t) = \int_0^t B(s)\mathrm{d}s$. Note that the trajectories of this non-Markovian process are smooth.  
     
	\item \textit{Ornstein-Uhlenbeck process} (OU). This is a Gaussian process $\{X(t)\}$ which satisfies the stochastic differential equation $\mathrm{d}X(t) = \theta (\mu - X(t))\mathrm{d}t + \sigma \mathrm{d}B(t)$. In our simulations we have chosen $\theta=\mu=\sigma=1$.
	
	\item \textit{Fractional Brownian Motion} (fBM). This process is an generalization of the Brownian motion $B(t)$ but, unlike $B(t)$, it has not independent increments. The mean function of this Gaussian process is identically 0 and its covariance function is $k(t,s) = \frac{1}{2}(|t|^{2H} + |s|^{2H}-|t-s|^{2H})$, where $H\in (0,1)$ is the so-called Hurst exponent. Note that for $H=0.5$, this process coincides with the standard Brownian Motion. Also, the trajectories of this process are still not differentiable at every point but the index $H$ is closely related to the H\"older continuity properties of these trajectories. This entails when $H>0.5$, the trajectories look ``more regular'' than those of the Brownian motion, having a wilder appearance for $H<0.5$. To cover both cases we have used $H=0.2$ and $H=0.8$ in our simulations.
\end{enumerate}

Figure \ref{Fig_TraySim} shows some trajectories of each of these six processes, where the variables $X(t)$ for $t$ in a neighbourhood of 0 have been omitted to satisfy the non-degeneracy requirements of the method. 

As for the regression function $g$ we have considered the following three choices.

\begin{figure}
\centering
\includegraphics[scale=0.42]{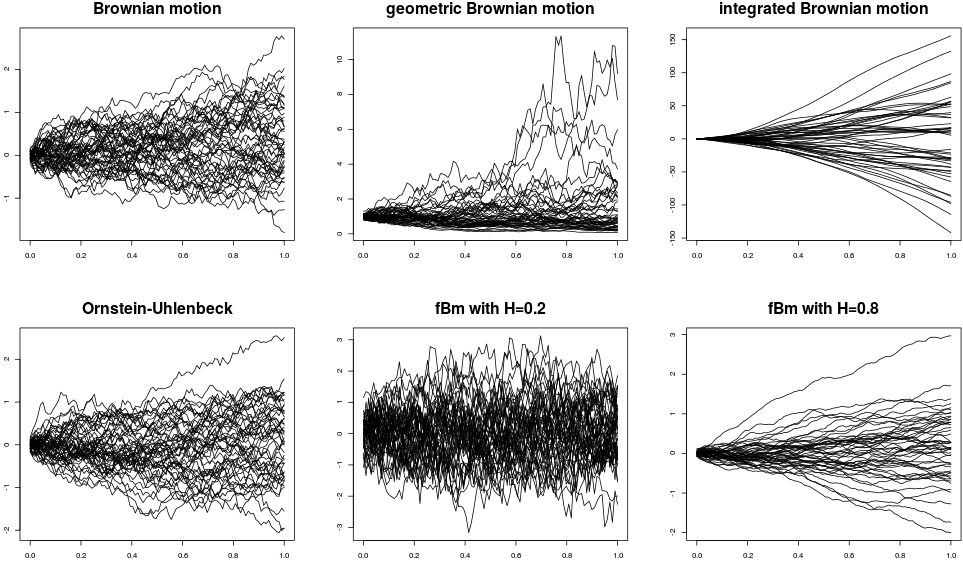}
\caption{\footnotesize 50 trajectories of each of the processes used in the simulations.}
\label{Fig_TraySim}
\end{figure}

\begin{enumerate}
	\item Two functions $\beta$ in \eqref{Eq:RKHSmodel} of type $\beta(t)=\sum_j\beta_jK(t_j^*,\cdot)$ so that the regression model reduces to the ``sparse version'' \eqref{Eq:RKHSmodel-finito}. We have considered two different regression functions. For the first one, we have used the set of points $T^* = (0.2, 0.4, 0.9) $ with weights $\beta_{T^*} = (2,-5,1)$; this is ``Regression model 1'' in the tables. For the second one, we have used $T^* = (0.16, 0.47, 0.6, 0.85, 0.91)$ and $\beta_{T^*} = (2.1, -0.2, -1.9, 5, 4.2)$ (``Regression model 2'' in the tables). Therefore, the response variables in both cases are given, respectively, by
	{\small \begin{eqnarray*}
	Y_1 &=& 2X(0.2) -5X(0.4) + X(0.9) + \varepsilon, \\
	Y_2 &=& 2.1X(0.16) -0.2X(0.47) -1.9X(0.67)+ 5X(0.85)+ 4.2X(0.91) + \varepsilon.
	\end{eqnarray*}}
	
	\item $\beta(t)=\log(1+t)$ and the regression model is \eqref{Eq:L2} with $\alpha_0=0$. Thus, the sparse RKHS model \eqref{Eq:RKHSmodel-finito} does not hold in this case.  This is ``Regression model 3'' in the tables. It has been already used in \cite{cuevas2002}. Therefore, the corresponding response variable is generated by
	$$Y_3 = \int_0^1 \log(1+t)X(t)\mathrm{d}t + \varepsilon.$$ 
\end{enumerate}

%%%%%%%%%%%%%%%%%%%%%%%%%%%%%%%%%%%%%%%%%%%%%%%%%%%%%%%%%%%%%%%%%%%%%%%
%%%%%%%%%%%%%% 				Real data			   %%%%%%%%%%%%%%%%%%
%%%%%%%%%%%%%%%%%%%%%%%%%%%%%%%%%%%%%%%%%%%%%%%%%%%%%%%%%%%%%%%%%%%%%%%

\subsection{Real data}
We have also checked the different methods when applied to two real data sets. Since these data have been already considered in other recent papers of the FDA literature, we will give only brief descriptions of them.
\begin{enumerate}
	\item \textit{Ash content in sugar samples}. This data set has been  used, for example, in  \cite{aneiros2014}. The version we use corresponds in fact to a subset of the whole data set, available in  \url{http://www.models.kvl.dk/Sugar_Process}.
	The response variable $Y$ is the percentage of ash content in $266$ sugar samples. The trajectories $X(t)$ are the fluorescence spectra from $275$ to $560$ nm at excitation wavelength $290$. These curves are discretized in a grid of $100$ equispaced points. 
	\item \textit{Mediterranean fruit flies}. This data set has been recently used, for example, in \cite{muller2016}. The trajectories $X(t)$ are the egg-laying profiles (number of eggs at the day $t$) for $512$ female Mediterranean fruit flies during a period of $30$ days, and the response variable $Y$ is the total number of eggs laid during their remaining lifetime. Therefore, the trajectories are discretized in a grid of $30$ points.  
\end{enumerate}

Figure \ref{Fig_TrayReal} shows some trajectories of both data sets.

\begin{figure}
\centering
\includegraphics[scale=0.4]{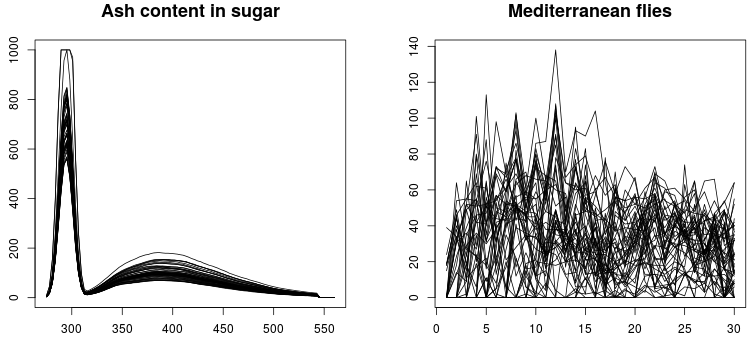}
\caption{\footnotesize 50 trajectories of each of the real data sets.}
\label{Fig_TrayReal}
\end{figure}

%%%%%%%%%%%%%%%%%%%%%%%%%%%%%%%%%%%%%%%%%%%%%%%%%%%%%%%%%%%%%%%%%%%%%%%
%%%%%%%%%%%%%% 				Methodology				   %%%%%%%%%%%%%%%%%%
%%%%%%%%%%%%%%%%%%%%%%%%%%%%%%%%%%%%%%%%%%%%%%%%%%%%%%%%%%%%%%%%%%%%%%%

\subsection{Methods under study and methodology}

We compare our proposal with other methods for variable selection recently considered in the literature which, in general, have been implemented in the original papers.  We now list the methods under study along with the notation used in the tables below.

\begin{enumerate}
	\item The method proposed in this paper (RKHS). It has been implemented using the iterative approximation described in equation \eqref{Eq:IterativeQn}.  The number of relevant points is chosen as explained in Section \ref{Sec:NumberP}. Therefore, no validation technique is required.	
	\item The variable selection procedure proposed in \cite{kneip2015} (KPS): in the original article, a mixed method for standard functional linear regression and variable selection technique is proposed. Since we are here concerned with variable selection, we have implemented just the corresponding part of the proposal. Essentially, the idea is to select the points (called ``impact points'' in \cite{kneip2015}) maximizing the covariance between the response variable $Y$ and a ``decorrelated'' version, $Z(t)$, of the original process. By construction, the decorrelated process $Z(t)$ is such that $Z(t)$ and $Z(s)$ are almost uncorrelated whenever $|t-s|\geq \delta$. The value $\delta$ and the number of selected variables are chosen using the BIC criterion, as proposed in the original paper.
	
	\item \textit{Partitioning Variable Selection} (PVS) with ML penalization, as proposed in \cite{aneiros2014}.  The original sample must be split into two independent subsamples, which should be asymptotically of the same sizes. The basic idea is to apply some multivariate variable selection technique in this context, but taking advantage of the functional structure of the data. The procedure works in two steps. In the first step, one constructs a equispaced subgrid of variables among all the variables in the original grid. Then a variable selection technique for multivariate data is applied on this subgrid, using the first subsample of the original data. For instance, we might use LASSO with ML penalization, as proposed in the original paper. Then, in the second step, this variable selection technique is applied again to an enlarged grid, constructed by taking all variables in an interval around those selected in step 1 (using the second subsample). Since the version of LASSO that we are considering selects automatically the number of variables, the only smoothing parameter to be selected here is the grid step for the points used in first stage of the method.  This parameter is set by cross-validation.
	
	\item \textit{Maxima Hunting} (MH), proposed in \cite{berrendero2016}: The original method was proposed for variable selection in supervised classification, but there is no conceptual restriction to apply the same procedure in a regression setting. The basic idea is to select the local maxima of the ``distance covariance” (association measure for random variables introduced in \cite{szekely2007}) between the response and the marginal variables of the process. In practice, the numerical identification of these maxima depends on a smoothing parameter $h$ which is chosen by cross-validation. The number of variables is also set by cross-validation. 
	
	\item \textit{Partial Least Squares} (PLS). This technique is well-\-known among the FDA users of functional data. The goal of PLS is not to pick up a few variables but to select some appropriate linear functionals of the original data (very much in the spirit of principal components analysis). So PLS is not a variable selection procedure, but a dimension reduction method. This means that PLS is not directly comparable to  the variable selection methods considered here, since its aims are not the exactly same. When we choose to use a variable selection procedure, it is understood that we want to perform some kind of dimension reduction still keeping the interpretability of the information directly given in terms of the original variables. By contrast, PLS might perhaps provide some gains in efficiency but at the expense of doing a dimensionality reduction with a more difficult interpretation. Anyway, we have included PLS in our study as a useful reference, just in order to check how much do we lose by restricting ourselves to variable selection methods. We have used the function ``fregre.pls.cv'' of the fda.usc R-package to compute the predictions.
	
	\item \textit{Base}. We denote by ``base'' the prediction methodolgy  derived from the standard $L_2$ linear regression model \eqref{Eq:L2}. No variable selection or dimension reduction procedure is done. Hence,  this method is incorporated just for the sake of comparison, to assess the accuracy of the predictions based on some previous variable selection procedure with those provided by the standard functional regression model \eqref{Eq:L2}.   We have used the function ``fregre.basis.cv'' of the fda.usc R-package (which relies on a basis representation of the trajectories) to compute the predictions.
\end{enumerate}

For each model, all methods are checked for the sample size $n = 150$, which has been split on 100 observations used as training data and 50 employed as test data.  As usual, the functional simulated data are discretized to $(X(t_1), \ldots, x(t_{100}))$, where $t_i$ are equispaced points in $[0, 1]$, starting from $t_1=1/100$. For all the experiments we obtain the Relative Mean Squared Error (RMSE) of each method, as defined by
\begin{equation} \label{Eq:RMSE}
RMSE\hspace*{1mm} (\widehat{Y};Y) = \frac{\sum_{i=1}^n \left(\widehat{Y}_i - Y_i \right)^2}{\sum_{i=1}^n Y_i^2}.
\end{equation}

Moreover, in those cases where $\beta$ has a ``sparse'' form of type $\beta(t)=\sum_j\beta_jK(t_j,\cdot)$ we obtain two measures of the accuracy in the variable selection procedure. Namely, we calculate the Hausdorff distance between the set of estimated points and the set of the real ones $T^*$ as well as the number of points selected ($\widehat{p}$) in order to compare it with the real number ($p^*$). We have imposed a maximum of $10$ selected points to the methods, except to PVS, since this method does not permit to decide the number of selected variables. The Hausdorff distance gives us an idea of the precision of the method since it takes into account the separation to the real points as well as the number of selected points. Each experiment has been replicated $100$ times.

%%%%%%%%%%%%%%%%%%%%%%%%%%%%%%%%%%%%%%%%%%%%%%%%%%%%%%%%%%%%%%%%%%%%%%%
%%%%%%%%%%%%%% 			Numerical Results			   %%%%%%%%%%%%%%%%%%
%%%%%%%%%%%%%%%%%%%%%%%%%%%%%%%%%%%%%%%%%%%%%%%%%%%%%%%%%%%%%%%%%%%%%%%

\subsection{Numerical outputs}

Tables \ref{Table_RMSEsim12}, \ref{Table_RMSEsim3} and \ref{Table_RMSEreal} provide the performance of the methods measured in terms of the Relative Mean Squared Error (Eq. \eqref{Eq:RMSE}) of the predictions, for each of the three regression models tested and for the two real data sets. Methods are presented in columns. Models appear in rows.  Table \ref{Table_Haus} contents the Hausdorff distance between the selected and the ``true'' relevant points for the four variable selection methods and the two models with sparse $\beta$ function. For these same experiments, Table \ref{Table_P} provides the number of selected points. In all the tables of the simulated data appear the mean and the standard deviation of each of the measured quantities.

In view of these tables, for the simulated data the proposed method seems to outperform the other variable selection procedures, according to all the considered criteria metrics (RMSE, Hausdorff distance and number of points) whenever a sparse model of type \eqref{Eq:L2VS} holds. The PVS method also performs quite well. When the model is not satisfied, and the response variable depends on the whole trajectory, PLS is the best in general, as expected. However, the order of magnitude of the error for PLS is the same as that of the variable selection methods in most cases. That is, using a few number of variables instead of the whole trajectory does not significantly affect  the prediction error, even if the response depends on the whole trajectory. For the sugar data, our method is slightly outperformed by the other variable selection ones, but it is better than the methods that use the whole trajectories. In addition, RKHS method uses only 5 points in this case, in contrast with KPS (10 points, which is the fixed maximum) and PVS (25 points, since it is not possible to select the number of points for this method). For the Mediterranean fruit flies set, the errors of all the variable selection techniques are very similar to the ones obtained with PLS and the base method.

%%%%%%%%%%%%%%%%%%%%%%%%%%%%%%%%%%%%%
%%%%% TABLAS DE RMSE SIM
%%%%%%%%%%%%%%%%%%%%%%%%%%%%%%%%%%%%%
\begin{table}
\centering
\resizebox{\textwidth}{!}{\begin{tabular}{l | l l l l l l }
%%%%%% fY1
\multicolumn{7}{c}{Regression model 1} \\\hline
&	RKHS  &  KPS   &  PVS   &  MH  &   PLS   &   Base \\ \hline
Bm &	\boxed{0.982} \hspace*{1mm}(0.264) &	2.14 \hspace*{1mm}(1.73) &	\boxed{1.14} \hspace*{1mm}(0.299) & 5.32 \hspace*{1mm}(2.92) & 5.27 \hspace*{1mm}(1.67) & 6.16 \hspace*{1mm}(1.56) \\
gBm &	\boxed{0.451} \hspace*{1mm}(0.487) &  5.15 \hspace*{1mm}(4.55) &  \boxed{0.61} \hspace*{1mm}(0.408) &  7.34 \hspace*{1mm}(11.6) &  4.59 \hspace*{1mm}(2.11) &  5.26 \hspace*{1mm}(2.28) \\
iBM & \boxed{0.017} \hspace*{1mm}(0.0205) &  0.0856 \hspace*{1mm}(0.106) &  0.13 \hspace*{1mm}(0.051) &  47.7 \hspace*{1mm}(13.4) &  0.129 \hspace*{1mm}(0.0473) &  \boxed{0.0305} \hspace*{1mm}(0.00963) \\
OU & 	\boxed{1.07} \hspace*{1mm}(0.315) &  2.2 \hspace*{1mm}(1.4) &  \boxed{1.2} \hspace*{1mm}(0.369) &  5.58 \hspace*{1mm}(3.11) &  5.51 \hspace*{1mm}(1.68) &  6.36 \hspace*{1mm}(1.72) \\
fBm 0.2 & \boxed{0.418} \hspace*{1mm}(0.117) &  1.94 \hspace*{1mm}(2.14) &  \boxed{0.539} \hspace*{1mm}(0.163) &  13.8 \hspace*{1mm}(4.2) &  20.8 \hspace*{1mm}(7.18) &  31.2 \hspace*{1mm}(8.71) \\
fBm 0.8 & \boxed{3.11} \hspace*{1mm}(0.995) &  3.52 \hspace*{1mm}(1.19) &  \boxed{3.22} \hspace*{1mm}(1.03) &  12.4 \hspace*{1mm}(7.73) &  4.11 \hspace*{1mm}(1.22) &  4 \hspace*{1mm}(1.21) \\
%%%%%% fY2
\multicolumn{7}{c}{} \\ %Linea en blanco
\multicolumn{7}{c}{Regression model 2} \\\hline
&	RKHS  &  KPS   &  PVS   &  MH  &   PLS   &   Base \\ \hline
Bm &	\boxed{0.0908} \hspace*{1mm}(0.0366) &  1.13 \hspace*{1mm}(0.658) &  \boxed{0.175} \hspace*{1mm}(0.0804) &  2.44 \hspace*{1mm}(1.06) &  0.757 \hspace*{1mm}(0.271) &  0.736 \hspace*{1mm}(0.239) \\
gBm & \boxed{0.0176} \hspace*{1mm}(0.00845) &  1.11 \hspace*{1mm}(1.01) &  \boxed{0.235} \hspace*{1mm}(0.259) &  2.31 \hspace*{1mm}(1.38) &  0.495 \hspace*{1mm}(0.267) &  0.566 \hspace*{1mm}(0.32) \\
iBm & \boxed{0.000271} \hspace*{1mm}(0.000311) &  0.00134 \hspace*{1mm}(0.000918) &  0.112 \hspace*{1mm}(0.0392) &  0.0519 \hspace*{1mm}(0.0149) &  0.00437 \hspace*{1mm}(0.00173) &  \boxed{0.000215} \hspace*{1mm}(0.0000669) \\
OU &	\boxed{0.0969} \hspace*{1mm}(0.0405) &  1.18 \hspace*{1mm}(0.642) &  \boxed{0.177} \hspace*{1mm}(0.07) &  2.47 \hspace*{1mm}(1.05) &  0.722 \hspace*{1mm}(0.239) &  0.681 \hspace*{1mm}(0.222) \\
fBm 0.2 & \boxed{0.0862} \hspace*{1mm}(0.021) &  0.48 \hspace*{1mm}(3.58) &  \boxed{0.161} \hspace*{1mm}(0.0622) &  6.49 \hspace*{1mm}(4.7) &  6.07 \hspace*{1mm}(1.92) &  7.77 \hspace*{1mm}(2.52) \\
fBm 0.8 & \boxed{0.0985} \hspace*{1mm}(0.0301) &  0.161 \hspace*{1mm}(0.0655) &  0.208 \hspace*{1mm}(0.107) &  0.416 \hspace*{1mm}(0.133) &  0.152 \hspace*{1mm}(0.0505) &  \boxed{0.125} \hspace*{1mm}(0.0387) \\
\end{tabular}
}
\caption{\footnotesize Mean and standard deviation of the RMSE for the response variable for simulated data sets with models 1 and 2 (scale of $10^{-2}$).}
\label{Table_RMSEsim12}
\end{table}

%%%%%% fY3
\begin{table}
\centering
\resizebox{\textwidth}{!}{\begin{tabular}{l | l l l l l l }
\multicolumn{7}{c}{Regression model 3} \\\hline
&	RKHS  &  KPS   &  PVS   &  MH  &   PLS   &   Base \\ \hline
Bm &	4.69 \hspace*{1mm}(1.39) & 4.41 \hspace*{1mm}(1.1) &  4.2 \hspace*{1mm}(1.24) &  \boxed{3.95} \hspace*{1mm}(1.02) & \boxed{3.84} \hspace*{1mm}(1.16) & 3.96 \hspace*{1mm}(1.11) \\
gBm &	1.46 \hspace*{1mm}(1.81) &  0.946 \hspace*{1mm}(0.296) &  1.06 \hspace*{1mm}(0.386) &  1.07 \hspace*{1mm}(0.462) &  \boxed{0.824} \hspace*{1mm}(0.252) &  \boxed{0.855} \hspace*{1mm}(0.291) \\
iBm &	0.00372 \hspace*{1mm}(0.00115) &  0.00337 \hspace*{1mm}(0.00105) &  0.0119 \hspace*{1mm}(0.0366) &  0.0416 \hspace*{1mm}(0.0105) & \boxed{0.00295} \hspace*{1mm}(0.000894) & \boxed{0.003} \hspace*{1mm}(0.000876) \\
OU &	4.8 \hspace*{1mm}(1.46) &  4.48 \hspace*{1mm}(1.21) &  4.25 \hspace*{1mm}(1.14) &  \boxed{3.98} \hspace*{1mm}(1.04) &  \boxed{3.88} \hspace*{1mm}(0.984) &  4 \hspace*{1mm}(0.988) \\
fBm 0.2 & 4.24 \hspace*{1mm}(1.08) &  4.29 \hspace*{1mm}(0.997) &  3.99 \hspace*{1mm}(0.864) &  3.84 \hspace*{1mm}(0.923) &  \boxed{3.81} \hspace*{1mm}(1.04) &  \boxed{3.63} \hspace*{1mm}(0.847) \\
fBm 0.8 & 4.98 \hspace*{1mm}(1.15) &  4.45 \hspace*{1mm}(0.945) &  4.22 \hspace*{1mm}(0.939) &  \boxed{3.93} \hspace*{1mm}(0.883) &  \boxed{3.88} \hspace*{1mm}(0.858) &  4.08 \hspace*{1mm}(0.947) \\
\end{tabular}
}
\caption{\footnotesize Mean and standard deviation of the RMSE for the response variable for simulated data sets with model 3 (scale of $10^{-1}$).}
\label{Table_RMSEsim3}
\end{table}

%%%%%%%%%%%%%%%%%%%%%%%%%%%%%%%%%%%%%
%%%%% TABLAS DE RMSE REAL
%%%%%%%%%%%%%%%%%%%%%%%%%%%%%%%%%%%%%
\begin{table}
\centering
\begin{tabular}{l | l l l l l l }
&	RKHS  &  KPS   &  PVS   &  MH  &   PLS   &   Base \\ \hline
Ash content in sugar & 0.354 & \boxed{0.185} & \boxed{0.231} & 0.265 & 0.748 & 0.465 \\
Mediterranean fruit flies & \boxed{0.999} & 1.02 & 1.05 & \boxed{0.944} & 1.03 & 1 \\
\end{tabular}
\caption{\footnotesize RMSE for the response variable for real data sets.}
\label{Table_RMSEreal}
\end{table}

%%%%%%%%%%%%%%%%%%%%%%%%%%%%%%%%%%%%%
%%%%% TABLAS DE HAUS DIST
%%%%%%%%%%%%%%%%%%%%%%%%%%%%%%%%%%%%%
\begin{table}
\centering
%\resizebox{\textwidth}{!}{
\begin{tabular}{l | l l l l}
\multicolumn{5}{c}{Regression model 1} \\\hline
 & RKHS & KPS & PVS & MH \\ \hline
 Bm & $\boxed{0.105} \hspace*{1mm}(0.15)$  & $2.15 \hspace*{1mm}(0.35)$  & $1.67 \hspace*{1mm}(0.539)$  & $2.05 \hspace*{1mm}(0.44)$ \\
 gBm & $\boxed{0.123} \hspace*{1mm}(0.0194)$  & $2.18 \hspace*{1mm}(0.217)$  & $1.65 \hspace*{1mm}(0.494)$  & $2.07 \hspace*{1mm}(0.517)$ \\
 iBm & $0.991 \hspace*{1mm}(0.0514)$  & $2.15 \hspace*{1mm}(0.281)$  & $\boxed{0.903} \hspace*{1mm}(0.619)$  & $5.91 \hspace*{1mm}(0.849)$ \\
 OU & $\boxed{0.097} \hspace*{1mm}(0.162)$  & $2.12 \hspace*{1mm}(0.285)$  & $1.71 \hspace*{1mm}(0.411)$  & $2.10 \hspace*{1mm}(0.623)$ \\
 fBm 0.2 & $\boxed{0} \hspace*{1mm}(0)$  & $2.11 \hspace*{1mm}(0.4)$  & $1.73 \hspace*{1mm}(0.489)$  & $2.01 \hspace*{1mm}(0.362)$ \\
 fBm 0.8 & $\boxed{0.203} \hspace*{1mm}(0.208)$   & $2.14 \hspace*{1mm}(0.277)$  & $1.62 \hspace*{1mm}(0.537)$   & $3.03 \hspace*{1mm}(1.72)$  \\
%%%%%% fY2
\multicolumn{5}{c}{} \\ %Linea en blanco
\multicolumn{5}{c}{Regression model 2} \\\hline
 & RKHS & KPS & PVS & MH \\ \hline
 Bm & $1.29 \hspace*{1mm}(0.0294)$  & $1.29 \hspace*{1mm}(0.195)$  & $\boxed{0.977} \hspace*{1mm}(0.369)$  & $1.53 \hspace*{1mm}(0.826)$ \\
gBm & $1.24 \hspace*{1mm}(0.167)$  & $1.68 \hspace*{1mm}(1.2)$  & $\boxed{1.18} \hspace*{1mm}(4.48)$  & $1.43 \hspace*{1mm}(1.43)$ \\
iBm & $\boxed{0.672} \hspace*{1mm}(0.316)$  & $1.28 \hspace*{1mm}(0.175)$  & $5.79 \hspace*{1mm}(2.52)$  & $7.22 \hspace*{1mm}(0.0851)$ \\
OU & $\boxed{2.03} \hspace*{1mm}(0.0911)$  & $2.19 \hspace*{1mm}(0.22)$  & $\boxed{2.03} \hspace*{1mm}(0.0611)$  & $2.14 \hspace*{1mm}(0.386)$ \\
fBm 0.2 & $1.3 \hspace*{1mm}(0)$  & $1.28 \hspace*{1mm}(0.24)$  & $\boxed{1.02} \hspace*{1mm}(0.408)$  & $1.38 \hspace*{1mm}(0.428)$ \\
fBm 0.8& $\boxed{1.26} \hspace*{1mm}(0.182)$   & $1.27 \hspace*{1mm}(0.151)$   & $1.42 \hspace*{1mm}(0.834)$   & $2.45 \hspace*{1mm}(0.943)$ \\
\end{tabular}
%}
\caption{\footnotesize Mean and standard deviation of the Hausdorff distance to the actual relevant points (Scale of $10^{-1}$).}
\label{Table_Haus}
\end{table}

%%%%%%%%%%%%%%%%%%%%%%%%%%%%%%%%%%%%%
%%%%% TABLAS DE P
%%%%%%%%%%%%%%%%%%%%%%%%%%%%%%%%%%%%%
\begin{table}
\centering
%\resizebox{\textwidth}{!}{
\begin{tabular}{l | l l l l }
\multicolumn{5}{c}{Regression model 1 ($p^*=3$)} \\\hline
&	RKHS  &  KPS   &  PVS   &  MH \\ \hline
Bm & $\boxed{3.21} \hspace*{1mm}(0.409)$  &  $8.94 \hspace*{1mm}(1.6)$  &  $11.42 \hspace*{1mm}(4.163)$  &  $5.95 \hspace*{1mm}(1.546)$ \\
gBm & $\boxed{3.63} \hspace*{1mm}(0.774)$  &  $9.2 \hspace*{1mm}(1.25)$  &  $10.96 \hspace*{1mm}(4.1)$  &  $6.23 \hspace*{1mm}(1.78)$ \\
iBm & $6.27 \hspace*{1mm}(0.468)$  &  $9.66 \hspace*{1mm}(0.781)$  &  $11.61 \hspace*{1mm}(3.39)$  &  $\boxed{1.15} \hspace*{1mm}(0.359)$ \\
OU & $\boxed{3.22} \hspace*{1mm}(0.426)$  &  $8.64 \hspace*{1mm}(1.77)$  &  $12.13 \hspace*{1mm}(4.23)$  &  $5.84 \hspace*{1mm}(1.61)$ \\
fBm 0.2 & $\boxed{3} \hspace*{1mm}(0)$  &  $8.55 \hspace*{1mm}(1.85)$  &  $11.36 \hspace*{1mm}(4.3)$  &  $5.92 \hspace*{1mm}(1.61)$ \\
fBm 0.8 & $3.32 \hspace*{1mm}(0.529)$   &  $9.21 \hspace*{1mm}(1.31)$  &  $11.1 \hspace*{1mm}(4.29)$  &  $\boxed{2.88} \hspace*{1mm}(1.16)$ \\
%%%%%% fY2
\multicolumn{5}{c}{} \\ %Linea en blanco
\multicolumn{5}{c}{Regression model 2 ($p^*=5$)} \\\hline
&	RKHS  &  KPS   &  PVS   &  MH  \\ \hline
Bm & $\boxed{4.91} \hspace*{1mm}(0.771)$  &  $9.7 \hspace*{1mm}(0.706)$  &  $8.88 \hspace*{1mm}(2.16)$  &  $5.7 \hspace*{1mm}(1.85)$ \\
gBm & $\boxed{5.5} \hspace*{1mm}(1.22)$  &  $9.46 \hspace*{1mm}(1.16)$  &  $8.77 \hspace*{1mm}(2.5)$  &  $6.02 \hspace*{1mm}(1.73)$ \\
iBm & $\boxed{6.81} \hspace*{1mm}(1.09)$  &  $9.58 \hspace*{1mm}(0.878)$  &  $7.73 \hspace*{1mm}(3.74)$  &  $1.08 \hspace*{1mm}(0.273)$ \\
OU & $\boxed{5.06} \hspace*{1mm}(0.874)$  &  $9.48 \hspace*{1mm}(0.979)$  &  $9.23 \hspace*{1mm}(1.76)$  &  $5.96 \hspace*{1mm}(1.73)$ \\
fBm 0.2 & $\boxed{4} \hspace*{1mm}(0)$  &  $9.33 \hspace*{1mm}(1.49)$  &  $7.28 \hspace*{1mm}(2.02)$  &  $6.61 \hspace*{1mm}(2.22)$ \\
fBm 0.8 & $\boxed{5.23} \hspace*{1mm}(0.617)$ &  $9.66 \hspace*{1mm}(0.901)$ &  $10.61 \hspace*{1mm}(3.14)$ &  $3.53 \hspace*{1mm}(1.11)$ \\
\end{tabular}
%}
\caption{\footnotesize Mean and standard deviation of the number of selected points ($\widehat{p}$).}
\label{Table_P}
\end{table}

On the other hand, regarding with the execution time, we also provide a couple of results. We have measured the execution time of the six methods for the third regression model when the functional data are drawn from the Ornstein-Uhlenbeck process and the fractional Brownian motion with $H=0.8$. The results can be seen in Table \ref{Table_TE}. As we have already mentioned, the RKHS-based method does not require any validation step to determine the number of selected variables. Therefore, the execution time is significantly smaller than that of the other variable selection methods. We can also see that the execution time for the PVS method is much bigger than the others. Note however that this method has in general a good behaviour in terms of prediction error.

%%%%%%%%%%%%%%%%%%%%%%%%%%%%%%%%%%%%%
%%%%% TABLAS DE TIEMPO
%%%%%%%%%%%%%%%%%%%%%%%%%%%%%%%%%%%%%
\begin{table}
\centering
\resizebox{\textwidth}{!}{
\begin{tabular}{c | l l l l l l }
&	RKHS  &  KPS   &  PVS   &  MH  &   PLS   &   Base \\ \hline
OU & $\boxed{0.224} \hspace*{1mm}(0.051)$ & $3.054 \hspace*{1mm}(9.637)$ & $36.622 \hspace*{1mm}(12.462)$ & $1,247 \hspace*{1mm}(0.055)$ & $0.239 \hspace*{1mm}(0.062)$ & $\boxed{0.228} \hspace*{1mm}(0.051)$ \\
fBm 0.8 & $\boxed{0.244} \hspace*{1mm}(0.060)$  & $2.738 \hspace*{1mm}(0.319)$  & $23.531 \hspace*{1mm}(4.829)$  & $1.299 \hspace*{1mm}(0.083)$  & $0.246 \hspace*{1mm}(0.058)$  & $\boxed{0.235} \hspace*{1mm}(0.055)$  \\
\end{tabular}
}
\caption{\footnotesize Mean and standard deviation of the execution time.}
\label{Table_TE}
\end{table}

%%%%%%%%%%%%%%%%%%%%%%%%%%%%%%%%%%%%%%%%%%%%%%%%%%%%%%%%%%%%%%%%%%%%%%%
%%%%%%%%%%%%%%%%%%%%%%%%%%%%%%%%%%%%%%%%%%%%%%%%%%%%%%%%%%%%%%%%%%%%%%%
%%%%%%%%%%%%%% 				Conclusions				   %%%%%%%%%%%%%%%%%%
%%%%%%%%%%%%%%%%%%%%%%%%%%%%%%%%%%%%%%%%%%%%%%%%%%%%%%%%%%%%%%%%%%%%%%%
%%%%%%%%%%%%%%%%%%%%%%%%%%%%%%%%%%%%%%%%%%%%%%%%%%%%%%%%%%%%%%%%%%%%%%%

\section{Conclusions}

The RKHS approach we have introduced in this paper provides a natural framework for a formal unified theory of variable selection for functional data. The ``sparse'' models (those where the variable selection techniques are fully justified) appear as particular cases in this setup. As a consequence, it is possible to derive asymptotic consistency results as those obtained in the paper. Likewise, it is also possible to consider the problem of estimating the ``true'' number of relevant variables in a consistent way, as we do in Section 5. This is in contrast with other standard proposals for which the number of variables is previously fixed as an input, or it is determined using cross validation and other computationally expensive methods. Then, our proposal is more firmly founded in theory and, at the same time, provides a  much faster method in practice, which is important when dealing with large data sets. 

The empirical results we have obtained are encouraging. In short, according to our experiments, the RKHS-based method works better than other variable selection methods is those sparse models that fulfil the ideal theoretical conditions we need. In the non sparse model considered in the simulations, the RKHS method is slightly outperformed by other proposals (but still behaves reasonably). Finally, in the ``neutral'' field of real data examples the performance looks also satisfactory and competitive.

Last but not least, from a general, methodological point of view, this paper represents an additional example of the surprising usefulness of reproducing kernels in statistics. Additional examples can be found in \cite{berlinet2004},  \cite{hsing2015theoretical}, \cite{berrendero2017} or \cite{kad16}.

%%%%%%%%%%%%%%%%%%%%%%%%%%%%%%%%%%%%%%%%%%%%%%%%%%%%%%%%%%%%%%%%%%%%%%%
%%%%%%%%%%%%%%%%%%%%%%%%%%%%%%%%%%%%%%%%%%%%%%%%%%%%%%%%%%%%%%%%%%%%%%%
%%%%%%%%%%%%%% 					Proofs				   %%%%%%%%%%%%%%%%%%
%%%%%%%%%%%%%%%%%%%%%%%%%%%%%%%%%%%%%%%%%%%%%%%%%%%%%%%%%%%%%%%%%%%%%%%
%%%%%%%%%%%%%%%%%%%%%%%%%%%%%%%%%%%%%%%%%%%%%%%%%%%%%%%%%%%%%%%%%%%%%%%

\section{Some additional proofs}\label{Sec:Proofs}

%%%%%%%%%%%%%%%%%%%%%%%%%%%%%%%%%%%%%%%%%%%%%%%%%%%%%%%%%%%%%%%%%%%%%%%
%%%%%%%%%%%%%% 			Q_0 con la fraccion			%%%%%%%%%%%%%%%%%%
%%%%%%%%%%%%%%%%%%%%%%%%%%%%%%%%%%%%%%%%%%%%%%%%%%%%%%%%%%%%%%%%%%%%%%%
\subsection{Proposition \ref{Prop:IterativeQ0Matrices}}\label{Sec:ProofIterativeQ0Matrices}

We have to rewrite the expression $c_{T_{p+1}}'\Sigma_{T_{p+1}}^{-1}c_{T_{p+1}}$, where $p\geq 1$. We can write the matrix $\Sigma_{T_{p+1}}$ in block form as
\begingroup
\renewcommand*{\arraystretch}{1.5}
\begin{eqnarray*}
\Sigma_{T_{p+1}} &=& \left(
\begin{array}{c c c | c}
  & & & \mbox{cov}(X(t_1),X(t_{p+1})) \\
  &\Sigma_{T_p}& &\vdots \\
  & & & \mbox{cov}(X(t_p),X(t_{p+1})) \\ \hline
  \mbox{cov}(X(t_1),X(t_{p+1})) & \ldots & \mbox{cov}(X(t_p),X(t_{p+1})) & \mbox{cov}(X(t_{p+1}),X(t_{p+1})) 
\end{array}
\right) \\[1em]
%%%
&\equiv& \left( \begin{array}{c | c}
\Sigma_{T_p} & c_{T_p,p+1} \\ \hline
c_{T_p,p+1}' & \sigma_{p+1}^2
\end{array} \right).
\end{eqnarray*}
\endgroup
Then its inverse matrix is
\begingroup
\renewcommand*{\arraystretch}{1.5}
\begin{eqnarray*}
\Sigma_{T_{p+1}}^{-1} &=& \left( \begin{array}{c | c}
\Sigma_{T_p}^{-1} + \frac{1}{a}\Sigma_{T_p}^{-1}c_{T_p,p+1}c_{T_p,p+1}'\Sigma_{T_p}^{-1} & -\frac{1}{a}\Sigma_{T_p}^{-1}c_{T_p,p+1}\\[1mm] \hline
-\frac{1}{a}c_{T_p,p+1}'\Sigma_{T_p}^{-1} & \frac{1}{a}
\end{array} \right),
\end{eqnarray*}
\endgroup
where $a = \sigma_{p+1}^2 - c_{T_p,p+1}'\Sigma_{T_p}^{-1}c_{T_p,p+1}.$ We can also write the vector of covariances as
$$c_{T_{p+1}}' \hspace*{2mm}=\hspace*{2mm} (\mbox{cov}(X(t_1),Y), \ldots, \mbox{cov}(X(t_{p+1}),Y)) \hspace*{2mm}=\hspace*{2mm} (c_{T_p} \hspace*{1mm}\vert\hspace*{1mm} c_{p+1}).$$
Using this notation we can rewrite the original expression as follows,
\begin{eqnarray*}\label{Eq:FracInv}
c_{T_{p+1}}' \Sigma_{T_{p+1}}^{-1} c_{T_{p+1}}  
%&=& \left( \begin{array}{c | c}
%c_{T_p}' & c_{p+1}
%\end{array} \right) 
%\left( \begin{array}{c | c}
%\Sigma_{T_p}^{-1} + \frac{1}{a}\Sigma_{T_p}^{-1}c_{T_p,p+1}c_{T_p,p+1}'\Sigma_{T_p}^{-1} & -\frac{1}{a}\Sigma_{T_p}^{-1}c_{T_p,p+1} \\[1mm] \hline
%-\frac{1}{a}c_{T_p,p+1}'\Sigma_{T_p}^{-1} & \frac{1}{a}
%\end{array} \right) 
%\left(\begin{array}{c}
%c_{T_p} \\\hline
%c_{p+1}
%\end{array}\right) \\[1em]
%%%%%
%&=& \Big( c_{T_p}'\left[\Sigma_{T_p}^{-1} + \frac{1}{a}\Sigma_{T_p}^{-1}c_{T_p,p+1}c_{T_p,p+1}'\Sigma_{T_p}^{-1}\right] + c_{p+1}\frac{-1}{a}c_{T_p,p+1}'\Sigma_{T_p}^{-1} \Big)c_{T_p} \\&&
%+ \Big( c_{T_p}'\frac{-1}{a}\Sigma_{T_p}^{-1}c_{T_p,p+1} + c_{p+1}\frac{1}{a} \Big) c_{p+1} \\[1em]
%%%%
&=& c_{T_p}' \Sigma_{T_p}^{-1}c_{T_p} + \frac{1}{a} c_{T_p}'\Sigma_{T_p}^{-1}c_{T_p,p+1}c_{T_p,p+1}'\Sigma_{T_p}^{-1} c_{T_p} - \frac{c_{p+1}}{a}c_{T_p,p+1}' \Sigma_{T_p}^{-1}c_{T_p}  \\&&
- \frac{c_{p+1}}{a} c_{T_p}'\Sigma_{T_p}^{-1}c_{T_p,p+1} + \frac{c_{p+1}^2}{a} \\[1em]
%%%
%&=& c_{T_p}' \Sigma_{T_p}^{-1}c_{T_p} + \frac{1}{a}\left[ c_{T_p}'\Sigma_{T_p}^{-1}c_{T_p,p+1} \left( c_{T_p}'\Sigma_{T_p}^{-1}c_{T_p,p+1} \right)' \right.\\&&\left.
%- c_{p+1} c_{T_p}'\Sigma_{T_p}^{-1}c_{T_p,p+1} - c_{p+1}\left(c_{T_p}'\Sigma_{T_p}^{-1}c_{T_p,p+1}\right)' + c_{p+1}^2 \right] \\[1em]
%%%
&=& c_{T_p}' \Sigma_{T_p}^{-1}c_{T_p} + \frac{1}{a}\left[ \left( c_{T_p}'\Sigma_{T_p}^{-1}c_{T_p,p+1} \right)^2 - 2 c_{p+1} c_{T_p}'\Sigma_{T_p}^{-1}c_{T_p,p+1}  + c_{p+1}^2 \right] \\[1em]
%%%
%&=& c_{T_p}' \Sigma_{T_p}^{-1}c_{T_p} + \frac{1}{a} \left( c_{T_p}'\Sigma_{T_p}^{-1}c_{T_p,p+1} - c_{p+1} \right)^2 \\[1em]
%%%
&=& c_{T_p}' \Sigma_{T_p}^{-1}c_{T_p} + \frac{\left( c_{T_p}'\Sigma_{T_p}^{-1}c_{T_p,p+1} - c_{p+1} \right)^2}{\sigma_{p+1}^2 - c_{T_p,p+1}'\Sigma_{T_p}^{-1}c_{T_p,p+1}},
\end{eqnarray*}
since the product $c_{T_p}'\Sigma_{T_p}^{-1}c_{T_p,p+1}$ is actually a real number.

%%%%%%%%%%%%%%%%%%%%%%%%%%%%%%%%%%%%%%%%%%%%%%%%%%%%%%%%%%%%%%%%%%%%%%%
%%%%%%%%%%%%%% 		Q0 iterativo con varianzas		%%%%%%%%%%%%%%%%%%
%%%%%%%%%%%%%%%%%%%%%%%%%%%%%%%%%%%%%%%%%%%%%%%%%%%%%%%%%%%%%%%%%%%%%%%
\subsection{Proposition \ref{Prop:IterativeQ0Vars}}\label{Sec:ProofIterativeQ0Vars}

Using the notation of the statement, if $\widetilde{X}(t)$ is the centred process, we can rewrite the numerator of the quotient of Equation \eqref{Eq:IterativeQ0Matrices} as
\begin{eqnarray*}
\mbox{cov}\big( Y - Y_{T_p}, X(t_{p+1}) \big) &=&  c_{p+1} - \mbox{cov}\big(Y_{T_p}, X(t_{p+1}) \big) \\[1em]
%%%
&=& c_{p+1} - c_{T_p}'\Sigma_{T_p}^{-1} \mbox{cov}\big( \widetilde{X}(T_p), X(t_{p+1})\big) \\[1em] 
%%%
&=& c_{p+1} - c_{T_p}'\Sigma_{T_p}^{-1}c_{T_p,p+1},
\end{eqnarray*}
since the covariances are not affected by the centring. We can also write $\mbox{cov}\big( Y - Y_{T_p}, X(t_{p+1}) - X_{T_p}(t_{p+1}) \big)$ in the same way, since $(Y-Y_{T_p}) \perp X_{T_p}(t_{p+1})$. For the denominator,
\begin{eqnarray*}
&&\hspace*{-1cm}\mbox{var}\big( X(t_{p+1}) - X_{T_p}(t_{p+1})\big) \\[1em]
&=& \mbox{var}\big( X(t_{p+1}) \big) +  \mbox{var}\big(X_{T_p}(t_{p+1})\big) - 2 \mbox{cov}(X(t_{p+1}), X_{T_p}(t_{p+1})) \\[1em]
%%%
%&=& \sigma_{p+1}^2 + \Vert X_{T_p}(t_{p+1}) \Vert_{\mathcal{L}_2}^2 - 2\langle X(t_{p+1}), X_{T_p}(t_{p+1}) \rangle_{\mathcal{L}_2}\\[1em]
%%%
&=& \sigma_{p+1}^2 +  c_{T_p,p+1}'\Sigma_T^{-1}c_{T_p,p+1} - 2c_{T_p,p+1}'\Sigma_T^{-1}\mbox{cov}\left( X(t_{p+1}), \widetilde{X}(T_p) \right)\\[1em]
%%%
&=& \sigma_{p+1}^2 -  c_{T_p,p+1}'\Sigma_T^{-1}c_{T_p,p+1}.
\end{eqnarray*}	
From these two expressions the proposition follows straightforwardly.

\section*{Acknowledgements}
This work has been partially supported by Spanish Grant MTM2016-78751-P 
and the European Social Fund (Ayudas para contratos predoctorales para la formaci\'on de doctores 2015, Ministerio de Econom\'ia y Competitividad, Spain).

%%%%%%%%%%%%%%%%%%%%%%%%%%%%%%%%%%%%%%%%%%%%%%%%%%%%%%%%%%%%%%%%%%%%%%%
%%%%%%%%%%%%%%%%%%%%%%%%%%%%%%%%%%%%%%%%%%%%%%%%%%%%%%%%%%%%%%%%%%%%%%%
%%%%%%%%%%%%%% 				References			   %%%%%%%%%%%%%%%%%%
%%%%%%%%%%%%%%%%%%%%%%%%%%%%%%%%%%%%%%%%%%%%%%%%%%%%%%%%%%%%%%%%%%%%%%%
%%%%%%%%%%%%%%%%%%%%%%%%%%%%%%%%%%%%%%%%%%%%%%%%%%%%%%%%%%%%%%%%%%%%%%%

\bibliography{Refs}{}

\end{document}